\documentclass[sigconf, authorversion, language=english]{acmart}  

\usepackage{microtype}  

\graphicspath{{./figs/}}
\usepackage{subcaption}

\usepackage{tabularx}  

\usepackage{enumitem}

\usepackage{mathtools}

\newcommand*{\Reals}{\mathbb{R}}

\newcommand*{\Naturals}{\mathbb{N}}

\DeclarePairedDelimiterX\set[1]\lbrace\rbrace{\,#1\,}

\usepackage{tikz}
\usetikzlibrary{positioning}

\usepackage{algorithm}
\usepackage[noend]{algpseudocode}

\newcommand*{\LongState}[2]{\State\parbox[t]{\dimexpr\linewidth-\dimexpr\algorithmicindent*#1}{#2\strut}}  

\DeclareMathOperator{\Neut}{Neut}
\DeclareMathOperator*{\Agg}{agg}
\DeclareMathOperator*{\AggAvg}{avg}
\DeclareMathOperator*{\AggMin}{min}

\newcommand*{\MiscProb}[1]{\textsc{\small #1}}
\newcommand*{\NeutProb}[1]{\textsc{Neutrality}_{#1}}
\newcommand*{\Alg}[1]{\textsc{\small Approx#1}}

\newcommand{\stitle}[1]{\vspace{0.5mm}\noindent{\textbf{#1}}.}

\newcommand{\tightsection}[1]{\section{#1}}
\newcommand{\tightsubsection}[1]{\subsection{#1}}

\copyrightyear{2023}
\acmYear{2023}
\setcopyright{acmlicensed}
\acmConference[KDD '23]{Proceedings of the 29th ACM SIGKDD Conference on Knowledge Discovery and Data Mining}{August 6--10, 2023}{Long Beach, CA, USA}
\acmBooktitle{Proceedings of the 29th ACM SIGKDD Conference on Knowledge Discovery and Data Mining (KDD '23), August 6--10, 2023, Long Beach, CA, USA}
\acmPrice{15.00}
\acmDOI{10.1145/3580305.3599425}
\acmISBN{979-8-4007-0103-0/23/08}

\begin{document}

\title{Maximizing Neutrality in News Ordering}

\author{Rishi Advani}
\orcid{0000-0002-5522-0401}
\affiliation{%
  \institution{University of Illinois Chicago}
  \city{Chicago}
  \state{IL}
  \country{USA}
}
\email{radvani2@uic.edu}

\author{Paolo Papotti}
\orcid{0000-0003-0651-4128}
\affiliation{%
  \institution{EURECOM}
  \city{Biot}
  \country{France}
}
\email{papotti@eurecom.fr}

\author{Abolfazl Asudeh}
\orcid{0000-0002-5251-6186}
\affiliation{%
  \institution{University of Illinois Chicago}
  \city{Chicago}
  \state{IL}
  \country{USA}
}
\email{asudeh@uic.edu}

\begin{abstract}
The detection of fake news has received increasing attention over the past few years, but there are more subtle ways of deceiving one's audience. In addition to the content of news stories, their presentation can also be made misleading or biased.
In this work, we study the impact of the ordering of news stories on audience perception. We introduce the problems of detecting cherry-picked news orderings and maximizing neutrality in news orderings. We prove hardness results and present several algorithms for approximately solving these problems.
Furthermore, we provide extensive experimental results and present evidence of potential cherry-picking in the real world.
\end{abstract}





\maketitle

\tightsection{Introduction}\label{sec:intro}
Access to information is a hallmark of modern democracy and society. Many people rely on online news sources or social media to understand the problems facing their communities, stay informed on current events, and determine who they would like to represent them in government. As such, they often have to blindly trust that their news sources are providing them with accurate information and presenting it in an unbiased way. Media organizations can take advantage of this trust to push their own agendas and spread disinformation when it benefits them financially or politically. As people become more aware of the prevalence of misinformation online, news sources risk losing their credibility if they are caught spreading outright lies. But even if the information they provide is technically accurate, there are still ways in which they can inject bias into its presentation~\citep{hamborg2019automated}.

One way in which this can be done is through deceptive ordering of news stories in a broadcast or web page. For example, suppose two headlines are placed next to each other in a user's feed:
\begin{itemize}[leftmargin=*]
    \item ``Immigration rates are on the rise again''
    \item ``Crime rates in major cities have reached historic highs''
\end{itemize}

Viewing one headline may influence the user's opinion of the story corresponding to the second headline --- by affecting their belief in the veracity of the story, their stance (for or against) on the events in the story, or by inducing them to perceive a causal relationship when there is only correlation. We term this phenomenon ``opinion priming''. Viewing one headline primes\footnotemark{} the user to form a certain opinion when shown the second headline.

\footnotetext{The occurrence of priming has been extensively studied in related settings~\citep{Baumgartner12, Agadjanian22, Simonsohn13, Draws21, Damessie18}. We performed a user study (\S\ref{sec:exp:userstudy}) to confirm the existence of priming in our setting.}

In this particular case, the user may perceive a causal relationship between these two events, even though it is never explicitly stated by the news source. In reality, this correlation could be completely spurious, but a news organization with ulterior motives could use this psychological trick by placing the two stories next to each other to influence the views of its audience.

On the other hand, a socially responsible news corporation, or an organization auditing a less scrupulous corporation to hold them accountable, may seek to order news stories in a way that minimizes this risk of opinion priming. Alternatively stated, they may seek to \emph{maximize the neutrality of a news ordering}.

Acknowledging that there are other objectives at play as well, including profitability for the news corporation and relevance to the user, in this paper, we focus on maximizing neutrality and leave simultaneous optimization of all these objectives as an important direction for future work.

\stitle{Problem Novelty}
In this paper, we study news ordering neutrality\footnotemark{} from an algorithmic perspective.
While there has been extensive work in recent years on different aspects of news coverage selection bias~\citep{hocke1996determining, bourgeois2018selection, kozman2022selection, lazaridou2016identifying}, diversifying news recommendations~\citep{kitchens2020understanding, gharahighehi2021diversification, lunardi2019representing, bauer2021diversification}, and computational fact-checking~\citep{NakovCHAEBPSM21, guo2022survey, zhou2020survey}, to the best of our knowledge, this paper is \emph{the first to consider the impact of the ordering of news stories on neutrality}.

\footnotetext{This paper addresses one important technical piece of a larger socio-technical problem with many dimensions~\citep{hamborg2019automated}. We discuss related work in more detail in \S\ref{sec:related}.}

\stitle{Contributions}
Our contributions in this paper are as follows:
\begin{itemize}[leftmargin=*]
    \item We formalize the notion of news ordering and introduce the problems of (a) detecting cherry-picked news orderings and (b) maximizing neutrality in news orderings (\S\ref{sec:pre}).
    \item We present an algorithm to efficiently detect cherry-picked news orderings (\S\ref{sec:detection}). The algorithm uses random shuffling and tail inequalities to detect if the neutrality of the given ordering is significantly different from the mean.
    \item We study the problem of maximizing neutrality in news orderings. We prove results on the theoretical hardness of solving this problem and provide several approximation algorithms (\S\ref{sec:aggsum} and \S\ref{sec:aggmin}).
    Our algorithms make (non-trivial) connections to other problems such as max-weight matching~\citep{edmonds_1965} and max-weight cycle cover~\citep{blaser2005approximating}, by using them as subroutines in the algorithms.
    \item We introduce new variations of the fundamental maximum traveling salesman problem and propose algorithms that can be used to solve problems with a broad range of applications.
    In particular, we define the \MiscProb{PathMaxTSP} problem of finding a Hamiltonian path with maximum total weight in a graph.
    \item We conduct comprehensive experiments on real and synthetic datasets to validate our theoretical results (\S\ref{sec:experiments}).
    We were able to find \emph{potential evidence of cherry-picked orderings in the real world}, further motivating our study.
    In addition, our user study with over 50 participants confirms the existence of priming in our setting.
\end{itemize}
We conclude the paper with
a discussion of related work (\S\ref{sec:related}) and
directions for future research (\S\ref{sec:discussion}).
\tightsection{Problem Setup}\label{sec:pre}
Let $\mathbf{t}=\set{t_1,t_2,\dots, t_n}$ be a set of $n$ news stories to be presented by a news source.
Let $\mathbf{s}=\set{s_1,s_2,\dots, s_n}$ be a permutation of the integers from 1 to $n$ representing an ordering of those news stories: news story $t_i$ is presented in position $s_i$.

When news headlines are placed near each other, the user's opinion of one may be influenced by the other.
Our objective is twofold: we aim to detect when a news source has cherry-picked the ordering of its news stories, and to find the ordering that minimizes this risk of opinion priming.

To model this, we define a \emph{pairwise opinion priming (POP) function}  $C\colon \mathbf{t} \times \mathbf{t} \to \Reals$ that takes as input a pair\footnotemark{} of stories $(t_i, t_j)$, where $t_i\neq t_j$, and returns a real number in the range $[0,1]$.
An output of 1 indicates certainty that opinion priming will occur between two stories if they are in adjacent slots. 
An output of 0 indicates that no opinion priming will occur.
Note that the likelihood of opinion priming occurring for a particular individual is impacted by their own beliefs and mentality and may differ from that of another individual. Thus, we consider the incidence of opinion priming over a group of individuals.
More precisely, the function $C$ reflects the \emph{average} pairwise opinion priming over the audience.

\footnotetext{We could instead use \emph{ordered} pairs if we wished to model the POP function as being affected by the order of the two stories, and nearly all the results in this paper would still hold. More details can be found in the appendix.}

The values of $C$ can be determined in several ways. 
For example, a real audience's perception can be surveyed, an auditing agency can crowdsource answers to questions on opinion priming between pairs of news stories, or a domain expert can assign values based on their own judgment.
We use crowdsourcing in our user studies to estimate the values of $C$, confirming this method's feasibility in practice.
In this paper, we assume the values of $C$ are given as input. Thus, the problems and solutions proposed in this paper are agnostic to the choice of technique for determining $C$.

We also consider the distance between two stories in an ordering. As the distance increases, any opinion priming between the pair of stories will diminish accordingly; the audience will not form as strong an association if the stories are presented far apart from each other.
We define a \emph{decay function} $D \colon \Naturals \to \Reals$ that takes as input the distance between two distinct time slots and returns a real number in the range $[0,1]$ with $D(1)=1$ and $D$ monotonic.

Using the POP and decay functions, we can now define the pairwise neutrality of a pair of news stories.
\begin{definition}[Pairwise Neutrality] \label{def:pb}
Given a set of news stories $\mathbf{t}$, an ordering $\mathbf{s}$, a POP function $C$, and a decay function $D$, the \emph{pairwise neutrality} between distinct news stories $t_i$ and $t_j$ is defined as
$N_{i,j} =  1 - D(|s_j - s_i|) \cdot C(t_i,t_j)$.
\end{definition}

We now give an example to illustrate the concepts discussed so far. Suppose we have the following decay function.
\begin{equation}\label{eq:decay}
   D(d) =
   \begin{cases}
       1 & \text{if } d=1 \\
       0 & \text{otherwise}
   \end{cases}  
\end{equation}
This function treats pairs of headlines as having no risk of opinion priming if they are more than one position away from each other.

\begin{example}\label{ex-1}
Consider a set of news stories $\mathbf{t}=\set{t_1,t_2,t_3,t_4}$ with the following POP function $C$.
\[
    \begin{array}{|c|c|c|c|c|}
         \hline
         C & t_1 & t_2 & t_3 \\
         \hline
         t_2 & 0.1 & & \\
         \hline
         t_3 & 0.3 & 0.7 & \\
         \hline
         t_4 & 0.2 & 0.8 & 1 \\
         \hline
    \end{array}
\]
If we order the stories in $\mathbf{t}$ as $t_1, t_3, t_4, t_2$, then we have the following values for $\mathbf{s}$.
\vspace{-2mm}
\[
    \begin{array}{|c|c|c|c|c|}
        \hline
        s_1 & s_2 & s_3 & s_4 \\
        \hline
        1 & 4 & 2 & 3 \\
        \hline
    \end{array}
\]
For example, $s_3=2$ because $t_3$ is placed second in the ordering.

Using \autoref{eq:decay} for the decay function, the pairwise neutrality between $t_1$ and $t_2$, for example, is
\[
N_{1,2} = 1 - D(|s_2 - s_1|) \cdot C(t_1,t_2) = 1 - 0 \times 0.1 = 1 \,.
\]
The pairwise neutrality for all pairs of news stories is given below.
\[
    \begin{array}{|c|c|c|c|}
         \hline
         N & t_1 & t_2 & t_3 \\
         \hline
         t_2 & 1 & & \\
         \hline
         t_3 & 0.7 & 1 & \\
         \hline
         t_4 & 1 & 0.2 & 0 \\
         \hline
    \end{array}
\]
\end{example}

Using the notion of pairwise neutrality, we can now define neutrality for a whole news ordering. At a high-level, a news ordering is neutral if the pairwise neutrality between all pairs of news stories is ``high''.
More formally, we use Definition~\ref{def:n} to quantify neutrality in a news ordering.
For our purposes, an \emph{aggregation function} is any function that takes a set as input and returns a single real number in $[0,1]$ as output.

\begin{definition}[News Ordering Neutrality] \label{def:n}
Given a set of news stories $\mathbf{t}$, a POP function $C$, a decay function $D$, and an aggregation function $\Agg$, the neutrality of a news ordering $\mathbf{s}$ is defined as
\vspace{-2mm}
\begin{equation*}
    \Neut_{\Agg}(\mathbf{s}) = \Agg_{1\leq i<j\leq n} N_{i,j} \,,
\end{equation*}
where $N_{i,j}$ is the pairwise neutrality between $t_i$ and $t_j$.
\end{definition}
Analogously, for any aggregation function $\Agg$, we will denote the optimization problem of finding the ordering $\mathbf{s}$ that maximizes $\Neut_{\Agg}$ by $\NeutProb{\Agg}$.

We now define two aggregation functions that we will use throughout the paper.
\begin{definition}[Conditional Average Aggregation]
    Given the pairwise neutrality values $N_{i,j}$ for a set of news stories, an ordering $\mathbf{s}$, and a decay function $D$, the conditional average is defined as the average of the pairwise neutrality values over the support of $D$. I.e., if $D^+$ is the set of pairs $(i,j)$ where $D(|s_j - s_i|) > 0$, then the conditional average is the average of the neutrality values $N_{i,j}$ over all $(i,j) \in D^+$. If $D > 0$ for all inputs, then this is just a simple average. For brevity, we will refer to this function by ``$\AggAvg$''.
\end{definition}
\begin{definition}[Minimum Aggregation]
    The minimum aggregation function simply returns the minimum element in a set. We will refer to this function by ``$\AggMin$''.
\end{definition}

\begin{example}
Consider the same set of news stories $\mathbf{t}$, ordering $\mathbf{s}$, POP function $C$, and decay function $D$ from Example~\ref{ex-1}.
Using $\AggAvg$ as the aggregation function, we have
\[
\Neut_{\AggAvg}(\mathbf{s})
= (0.7 + 0.2 + 0)/3 = 0.3 \,.
\]
Similarly, the neutrality of $\mathbf{s}$ under $\AggMin$ aggregation is
\[
\Neut_{\AggMin}(\mathbf{s})
= \min N_{i,j}
= 0 \,.
\]
\end{example}

\begin{table}[pt]
\caption{Table of Notations}
\vspace{-3mm}
\small
\begin{tabularx}{\linewidth}{cX}
    \toprule
    Notation & Description \\
    \midrule
    $n$ & The cardinality of the set $\mathbf{t}$ \\
    $t_i$ & A news story in the set $\mathbf{t}$ \\
    $s_i$ & The slot assigned to $t_i$ in the ordering $\mathbf{s}$ \\
    $C$ & The pairwise opinion priming function \\
    $D$ & The decay function \\
    $N_{i,j}$ & The pairwise neutrality between $t_i$ and $t_j$ \\
    $\Neut_{\Agg}(\mathbf{s})$ & The neutrality of the ordering $\mathbf{s}$ under the aggregation function ``$\Agg$'' \\
    \bottomrule
\end{tabularx}
\label{tab:notations}
\end{table}

Having defined the notion of neutrality in news ordering, we will begin by studying how to detect cherry-picked news orderings in \S\ref{sec:detection}. Our main objective in this paper is to find news orderings that maximize neutrality, which we shall do in \S\ref{sec:aggsum} and \S\ref{sec:aggmin}. While the techniques proposed in \S\ref{sec:detection} are agnostic to the choice of decay function, in \S\ref{sec:aggsum} and \S\ref{sec:aggmin}, we will restrict ourselves to the decay function given in \autoref{eq:decay}. This allows us to model the problem using the language of graph theory. Analyzing more complex decay functions is an important direction for future work.

We define a graph representation of the problem as follows. For each news story $t_i$, we include a vertex $v_i$. For each pair of distinct stories $t_i$ and $t_j$, we include an edge between $v_i$ and $v_j$ with weight $N_{i,j}$.
For brevity, henceforth in this paper, assume all graphs are simple, complete, undirected, weighted, and have nonnegative edge weights unless otherwise specified. The requirement that the graphs are simple and complete is equivalent to stating that every pair of distinct vertices is joined by exactly one edge.





We define some graph theory terms that are used in the paper.
\begin{definition}[Hamiltonian cycle]
In a graph $G = (V,E)$, a \emph{Hamiltonian cycle} is a simple cycle that includes all vertices in $V$.
\end{definition}
\begin{definition}[Hamiltonian path]
In a graph $G = (V,E)$, a \emph{Hamiltonian path} is a simple path that includes all vertices in $V$.
\end{definition}
\begin{definition}[$\MiscProb{HamPath}$]
Given a graph $G$, \emph{$\MiscProb{HamPath}$} is the problem of determining if there exists a Hamiltonian path in $G$.
\end{definition}

With the restriction of the decay function to \autoref{eq:decay}, the problem of finding an ordering of news stories that maximizes $\Neut$ is equivalent to finding a Hamiltonian path that maximizes $\Neut$.
But first, in \S\ref{sec:detection} we propose our algorithm for detecting cherry-picking in an ordering (with any decay function).
\tightsection{Detecting Cherry-Picked Orderings}\label{sec:detection}
We begin by illustrating how to detect cherry-picked news orderings.
Suppose we have a set of news stories $\mathbf{t}$, a POP function $C$, a decay function $D$, and an aggregation function $\Agg$. Then, given a news ordering $\mathbf{s}$, we can deduce that it was likely cherry-picked if $\Neut_{\Agg}(\mathbf{s})$ differs significantly from the average neutrality over all possible orderings of $\mathbf{t}$.
If $\Neut_{\Agg}(\mathbf{s})$ is significantly lower than the average, then we have successfully detected bias in the ordering. On the other hand, if $\Neut_{\Agg}(\mathbf{s})$ is significantly higher than the average, then we can determine that the specified ordering was deliberately chosen in the interest of fairness.

If we knew the population mean and standard deviation, we could use Chebyshev's inequality to obtain an upper bound on the deviation from the mean. However, the number of possible orderings is combinatorially large ($n!$), so we cannot compute the neutrality for all of them. If we instead generate a sample of $r$ random orderings using Fisher-Yates shuffles~\citep{fisher1953statistical,10.1145/364520.364540}, we can use the Saw-Yang-Mo inequality~\citep{10.2307/2683249}, which only requires the \emph{sample} mean and standard deviation, to obtain an upper bound on deviation from the sample mean. For convenience, we use a simplified (and slightly looser) form of Kab{\'a}n's variant of the inequality~\citep{Kaban2012}:
\begin{equation}\label{eq:Kaban}
    \Pr\Biggl( \bigl\lvert X - \overline{Y} \bigr\rvert \geq \lambda \sigma \sqrt{\frac{r+1}{r}} \Biggr) \leq \frac{1}{\lambda^2} + \frac{1}{r} \,,
\end{equation}
where $\overline{Y}$ is the sample mean, $\sigma$ is the unbiased sample standard deviation,\!\footnotemark{} and the value for $\lambda$ is set such that the difference between the neutrality of the given ordering and $\overline{Y}$ is $\lambda \sigma \sqrt{(r+1)/r}$.

\footnotetext{Usually $\sigma$ is used for population deviation and $s$ for sample deviation, but we chose to avoid the use of $s$ to avoid confusion with our notation $\mathbf{s}$ for news orderings.}

\begin{example}
If we use, say, $r=50$ samples with $\lambda=5$, then using Equation~\ref{eq:Kaban}, we have the following.
\begin{equation*}\hspace{5mm}
    \Pr\Biggl( \bigl\lvert X - \overline{Y} \bigr\rvert \geq 5 \sigma \sqrt{\frac{51}{50}} \Biggr) \leq \frac{1}{5^2} + \frac{1}{50} = \frac{3}{50}
\end{equation*}
The probability that the neutrality of a truly random ordering is greater than $5 \sqrt{51/50} \approx 5.05$ sample standard deviations from the sample mean is less than $6\%$. Thus, if the neutrality of our given ordering is that far from the sample mean, it is highly likely that it was cherry-picked.
\end{example}

Users can select the value for parameter $r$ based on their problem size, aggregation function's complexity, access to computational resources, and error tolerance. We suggest $r=300$ as a reasonable starting point.


Now, we analyze the time complexity of the detection procedure.
We can compute $r$ random permutations in $O(rn)$ time using Fisher-Yates shuffles. We can compute the neutrality of the $r$ orderings in $O(rn^2)$ time. Computing the sample mean and standard deviation of the $r$ values takes $O(r)$ time and evaluating the test statistic takes $O(1)$ time.
Thus, overall, the algorithm takes $O(rn^2)$ time.

Furthermore, if we make certain assumptions, we can obtain a running time linear in $n$. If we use the decay function given by \autoref{eq:decay} along with an aggregation function that can be computed in linear time (e.g., $\AggAvg$ or $\AggMin$), we can compute the neutrality of the $r$ orderings in $O(rn)$ time for an overall running time of $O(rn)$.
\tightsection{Maximizing Neutrality under Average Aggregation}\label{sec:aggsum}
In the previous section, we considered the problem of detecting cherry-picked news orderings. Now, we move on to the main focus of our paper: finding news orderings with maximum neutrality.

Before beginning the technical content, we stop to emphasize the importance of computational approaches to this problem. Due to the combinatorially large number of possible orderings, the task is infeasible for a human with even a very small number of news headlines. For example, \emph{with \underline{\smash{only 10 headlines}}, there are 3,628,800 potential orderings to consider}. Thus, even in contexts where few stories are presented (e.g., a television broadcast), computational approaches are important. Furthermore, there are contexts in which the number of stories grows much larger (e.g., scrolling through a social media feed), where computational approaches are critical.

First, we consider the scenario where our aggregation function is the $\AggAvg$ function. This is a natural aggregation function to use; if we are equally invested in the pairwise neutrality of each pair of stories, it makes sense to maximize the average (mean) value. Note that this is exactly equivalent to maximizing the sum of the pairwise neutrality of each pair but with the added benefit that the neutrality will always be a value in the range $[0,1]$, so it is easier to make intuitive judgments about whether it is ``high'' or ``low''.

In the graph theory representation, the problem is now equivalent to finding a Hamiltonian path with maximum weight. To the best of our knowledge,
\emph{we are the first to study this problem.}
We will call this problem the ``path maximum traveling salesman problem'', or $\MiscProb{PathMaxTSP}$.
\begin{definition}[$\MiscProb{PathMaxTSP}$]
Given a graph $G$, \emph{$\MiscProb{PathMaxTSP}$} is the problem of finding a Hamiltonian path with maximum total weight.
\end{definition}

\begin{theorem}\label{thm:pathmaxtspnphard}
$\MiscProb{PathMaxTSP}$ is NP-hard.\footnotemark{}
\end{theorem}
\footnotetext{Proofs of all theorems stated in this section can be found in the appendix.}
\begin{corollary}\label{thm:avgnphard}
$\NeutProb{\AggAvg}$ is NP-hard.
\end{corollary}
Given the above hardness results, 
in the rest of the section, we design approximation algorithms to solve $\MiscProb{PathMaxTSP}$.

Between our algorithms \Alg{Mat} and \Alg{CC}, \Alg{CC} has improved efficiency with the same approximation factor, so we advocate for its use over \Alg{Mat} in all cases. We include \Alg{Mat} in our exposition for its comparative simplicity and in the hope that it inspires future work in this area. Our algorithm \Alg{3CC} achieves the best approximation factor but has an unreasonably slow running time.

\tightsubsection{Approximation via Iterated Matching}
The first algorithm, \Alg{Mat} (pseudocode in the appendix), works by making a connection to the well-known max-weight matching problem, where in a weighted graph, the goal is to find a set of disjoint edges with maximum total weight~\citep{kleinberg2006algorithm}.

In each iteration $k$, the algorithm constructs a graph $G_{k+1}$ used in the next iteration. To do so, it first finds a max-weight matching in the graph $G_k$.
Then, for every pair in the matching, it adds a ``super node'' to $G_{k+1}$. 
The super node represents a path in the original graph $G$. To perform this merge, the algorithm joins the represented paths by the pair of endpoints with maximum edge weight. If $|V_k|$ is odd, one of nodes in $G_k$ remains unmatched and gets added to $G_{k+1}$ as is.
The weight of the edge between each pair of nodes in $G_{k+1}$ is the maximum edge weight between the ends of their represented paths.
The algorithm continues this process until there is only one super node left.
The path represented by the final super node is a Hamiltonian path in the original graph.

\begin{figure*}[bhpt]
\small
\begin{subfigure}[t]{0.3\linewidth}
\centering
\begin{tikzpicture}
\tikzset{node distance = 1.5cm}
\tikzset{every node/.style={circle, draw}}
\tikzset{every path/.style={thick, -}}
\tikzset{special/.style={draw=red, line width=2pt}}
\node (t1) {$t_1$};
\node (t2) [right = of t1] {$t_2$};
\node (t3) [below = of t1] {$t_3$};
\node (t4) [below = of t2] {$t_4$};
\node (t5) [right = of t2] {$t_5$};
\node (t6) [right = of t4] {$t_6$};
\tikzset{every node/.style={circle, draw=white, fill=white}}
\draw (t1) -- node {$1$} (t2);
\draw[special] (t1) -- node {$0.8$} (t3);
\draw (t1) -- node [near start] {$0.1$} (t4);
\draw (t2) -- node [near start] {$1$} (t3);
\draw[special] (t2) -- node {$0.3$} (t4);
\draw (t3) -- node {$0$} (t4);
\draw (t2) -- node {$0$} (t5);
\draw (t4) -- node {$0.9$} (t6);
\draw[special] (t5) -- node {$1$} (t6);
\draw (t2) -- node [near start] {$0$} (t6);
\draw (t5) -- node [near start] {$1$} (t4);
\end{tikzpicture}
\vspace{-2mm}
\caption{\scriptsize First iteration}
\label{fig:ex2-1}
\end{subfigure}
\hfill
\begin{subfigure}[t]{0.3\linewidth}
\centering
\begin{tikzpicture}
\tikzset{node distance = 1.5cm}
\tikzset{every node/.style={circle, draw}}
\tikzset{every path/.style={thick, -}}
\tikzset{special/.style={draw=red, line width=2pt}}
\node (t1) {$t_1$};
\node (t2) [right = of t1] {$t_2$};
\node (t3) [below = of t1] {$t_3$};
\node (t4) [below = of t2] {$t_4$};
\node (t5) [right = of t2] {$t_5$};
\node (t6) [right = of t4] {$t_6$};
\tikzset{every node/.style={circle, draw=white, fill=white}}
\draw (t1) -- node {$1$} (t2);
\draw[special] (t1) -- node {$0.8$} (t3);
\draw (t1) -- node [near start] {$0.1$} (t4);
\draw (t2) -- node [near start] {$1$} (t3);
\draw[special] (t2) -- node {$0.3$} (t4);
\draw (t3) -- node {$0$} (t4);
\draw (t2) -- node {$0$} (t5);
\draw (t4) -- node {$0.9$} (t6);
\draw[special] (t5) -- node {$1$} (t6);
\draw (t2) -- node [near start] {$0$} (t6);
\draw[special] (t5) -- node [near start] {$1$} (t4);
\end{tikzpicture}
\vspace{-2mm}
\caption{\scriptsize Second iteration}
\label{fig:ex2-2}
\end{subfigure}
\hfill
\begin{subfigure}[t]{0.3\linewidth}
\centering
\begin{tikzpicture}
\tikzset{node distance = 1.5cm}
\tikzset{every node/.style={circle, draw}}
\tikzset{every path/.style={thick, -}}
\tikzset{special/.style={draw=red, line width=2pt}}
\node (t1) {$t_1$};
\node (t2) [right = of t1] {$t_2$};
\node (t3) [below = of t1] {$t_3$};
\node (t4) [below = of t2] {$t_4$};
\node (t5) [right = of t2] {$t_5$};
\node (t6) [right = of t4] {$t_6$};
\tikzset{every node/.style={circle, draw=white, fill=white}}
\draw[special] (t1) -- node {$1$} (t2);
\draw[special] (t1) -- node {$0.8$} (t3);
\draw (t1) -- node [near start] {$0.1$} (t4);
\draw (t2) -- node [near start] {$1$} (t3);
\draw[special] (t2) -- node {$0.3$} (t4);
\draw (t3) -- node {$0$} (t4);
\draw (t2) -- node {$0$} (t5);
\draw (t4) -- node {$0.9$} (t6);
\draw[special] (t5) -- node {$1$} (t6);
\draw (t2) -- node [near start] {$0$} (t6);
\draw[special] (t5) -- node [near start] {$1$} (t4);
\end{tikzpicture}
\vspace{-2mm}
\caption{\scriptsize Third (final) iteration}
\label{fig:ex2-3}
\end{subfigure}
\Description{Three graphs; the first has three disjoint paths highlighted, the second has two disjoint paths highlighted, and the third has a single Hamiltonian path highlighted.}
\vspace{-3mm}
\caption{(Example~\ref{ex-2}) \Alg{Mat} returns the ordering $\langle t_3, t_1, t_2, t_4,t_5,t_6\rangle$}
\label{fig:ex2}
\end{figure*}

\begin{example}\label{ex-2}
Consider a set of stories $\mathbf{t}=\set{t_1,\dots,t_6}$ with pairwise neutrality values as shown in the graph of Figure~\ref{fig:ex2-1} (for visual clarity, we omit four edges with weight zero).
\Alg{Mat} starts by finding the max-weight matching $\{(t_1,t_3), (t_2,t_4), (t_5,t_6)\}$, as highlighted in the figure.
Next, the algorithm replaces the pairs in the matching with super nodes: $\langle t_1,t_3\rangle$, $\langle t_2,t_4\rangle$, and $\langle t_5,t_6\rangle$.
In the second iteration, the algorithm selects edge $(t_4,t_5)$ with weight $1$ to join the super nodes $\langle t_2,t_4\rangle$ and $\langle t_5,t_6\rangle$ (Figure~\ref{fig:ex2-2}).
In the final iteration, the algorithm matches $\langle t_2,t_4,t_5,t_6\rangle$ to $\langle t_1,t_3\rangle$, via edge $(t_1,t_2)$, creating the final super node, $\langle t_3,t_1,t_2,t_4,t_5,t_6\rangle$.
The neutrality of the resulting ordering under $\AggAvg$ aggregation is $4.1 / 5 = 0.82$.
\end{example}

\begin{theorem}\label{th:mtsp1}
\Alg{Mat} returns a $1/2$-approximation for $\MiscProb{Path}$-$\MiscProb{MaxTSP}$.
\end{theorem}


We have not yet explicitly specified a subroutine to compute a max-weight matching. Classically, this can be done in $O\bigl(n^4\bigr)$ time using Edmonds' blossom algorithm~\citep{edmonds_1965}. Alternatively, it can be done in $O\bigl(n^2\varepsilon^{-1}\log\varepsilon^{-1}\bigr)$ time for any fixed error $\varepsilon$~\citep{10.1145/2529989}.
Properly implemented, the runtime of each iteration of the loop is dominated by the cost of the matching. If we use Edmonds' blossom algorithm, then each iteration takes $O\bigl(|V_k|^4\bigr)$ time. By the master theorem~\citep{clrs}, the overall runtime is then $O\bigl(n^4\bigr)$.
While polynomial, \Alg{Mat} has a high time complexity. Therefore, next we propose our algorithm \Alg{CC} that, while maintaining the same approximation ratio, reduces time complexity by a factor of $n$.

\tightsubsection{Approximation via Iterated Cycle Cover}
The second algorithm, \Alg{CC} (pseudocode in \autoref{alg:PathMaxTSP4}), finds a max-weight cycle cover, defined as a set of cycles\footnote{Here, cycles of length 2 are allowed.} of maximum total weight such that every vertex is included in exactly one cycle. Then, it removes the min-weight edge from each cycle. The resulting paths are treated as super nodes (as in \Alg{Mat}) and the process is repeated until there is only one super node remaining. The final super node implicitly gives a Hamiltonian path in the original graph.

\begin{algorithm}[bhpt]
	\caption{Approximating $\MiscProb{PathMaxTSP}$ via iterated cycle cover}
    \label{alg:PathMaxTSP4}
	\begin{algorithmic}[1]
	\Procedure{ApproxCC}{$G=(V,E)$}
	    \State $k \gets 0$; 
	     $G_k \gets G$;
	     $size \gets n$
	    \While{$size > 1$}
    		\State Compute a max-weight cycle cover $C$ in $G_k$.
    		\State Construct a graph $G_{k+1}=(V_{k+1},E_{k+1})$ as follows.
    		\State $V_{k+1} \gets \set{}$
    		\ForAll{$c \in C$}
    		    \LongState{3}{Remove the min-weight edge, and let $(p_1), \dots, (p_a)$ be the resulting path.}
    		    \State $p \gets p_1$
    		    \ForAll{$i \in [2,a]$}
    		        \LongState{4}{Let $u_1, \dots, u_b$ be the path denoted by $p$ and $v_1, \dots, v_d$ be the path denoted by $p_i$.}
    		      \State \textbf{if} $w(u_b,v_1) > w(u_b,v_d)$ \textbf{then} $p \gets p, v_1, \dots, v_d$ 
    		      \State \textbf{else} $p \gets p, v_d, \dots, v_1$
    		    \EndFor
    		    \State Add $(p)$ to $V_{k+1}$.
    		\EndFor
    		\ForAll{pairs $(u_1,\dots,u_a), (v_1,\dots,v_b)$ in $V_{k+1}$}
    		    \LongState{3}{Add the edge $((u_1,\dots,u_a), (v_1,\dots,v_b))$ to $E_{k+1}$ with weight $w(u_a,v_1)$.}  
    		\EndFor
    		\State $size \gets |V_{k+1}|$; 
    		 $k \gets k+1$
	    \EndWhile
        \LongState{1}{\textbf{return} the Hamiltonian path $v_1, \dots, v_n$ in $G$ where $(v_1, \dots, v_n)$ is the sole vertex in $V_k$.}
	\EndProcedure
	\end{algorithmic}
\end{algorithm}

In this algorithm, we reduce the problem of computing a cycle cover to that of computing a bipartite matching. The original reduction is due to \citet{tutte_1954}; an accessible presentation of the specific case we are interested in is given by \citet{Nikolaev2021}.

\begin{figure*}[bhpt]
\small
\begin{subfigure}{0.3\linewidth}
\centering
\begin{tikzpicture}
\tikzset{node distance = 1.5cm}
\tikzset{every node/.style={circle, draw}}
\tikzset{every path/.style={thick, -}}
\tikzset{special/.style={draw=red, line width=2pt}}
\node (t1) {$t_1$};
\node (t2) [right = of t1] {$t_2$};
\node (t3) [below = of t1] {$t_3$};
\node (t4) [below = of t2] {$t_4$};
\node (t5) [right = of t2] {$t_5$};
\node (t6) [right = of t4] {$t_6$};
\tikzset{every node/.style={circle, draw=white, fill=white}}
\draw[special]  (t1) -- node {$1$} (t2);
\draw[special] (t1) -- node {$0.8$} (t3);
\draw (t1) -- node [near start] {$0.1$} (t4);
\draw[special] (t2) -- node [near start] {$1$} (t3);
\draw (t2) -- node {$0.3$} (t4);
\draw (t3) -- node {$0$} (t4);
\draw (t2) -- node {$0$} (t5);
\draw[special]  (t4) -- node {$0.9$} (t6);
\draw[special]  (t5) -- node {$1$} (t6);
\draw (t2) -- node [near start] {$0$} (t6);
\draw [special] (t5) -- node [near start] {$1$} (t4);
\end{tikzpicture}
\vspace{-2mm}
\caption{\scriptsize First iteration}
\label{fig:ex3-1}
\end{subfigure}
\hfill
\begin{subfigure}{0.3\linewidth}
\centering
\begin{tikzpicture}
\tikzset{node distance = 1.5cm}
\tikzset{every node/.style={circle, draw}}
\tikzset{every path/.style={thick, -}}
\tikzset{special/.style={draw=red, line width=2pt}}
\node (t1) {$t_1$};
\node (t2) [right = of t1] {$t_2$};
\node (t3) [below = of t1] {$t_3$};
\node (t4) [below = of t2] {$t_4$};
\node (t5) [right = of t2] {$t_5$};
\node (t6) [right = of t4] {$t_6$};
\tikzset{every node/.style={circle, draw=white, fill=white}}
\draw[special]  (t1) -- node {$1$} (t2);
\draw (t1) -- node {$0.8$} (t3);
\draw (t1) -- node [near start] {$0.1$} (t4);
\draw[special] (t2) -- node [near start] {$1$} (t3);
\draw (t2) -- node {$0.3$} (t4);
\draw (t3) -- node {$0$} (t4);
\draw (t2) -- node {$0$} (t5);
\draw  (t4) -- node {$0.9$} (t6);
\draw[special]  (t5) -- node {$1$} (t6);
\draw (t2) -- node [near start] {$0$} (t6);
\draw [special] (t5) -- node [near start] {$1$} (t4);
\end{tikzpicture}
\vspace{-2mm}
\caption{\scriptsize First iteration (transforming cycles to paths)}
\label{fig:ex3-2}
\end{subfigure}
\hfill
\begin{subfigure}{0.3\linewidth}
\centering
\begin{tikzpicture}
\tikzset{node distance = 1.5cm}
\tikzset{every node/.style={circle, draw}}
\tikzset{every path/.style={thick, -}}
\tikzset{special/.style={draw=red, line width=2pt}}
\node (t1) {$t_1$};
\node (t2) [right = of t1] {$t_2$};
\node (t3) [below = of t1] {$t_3$};
\node (t4) [below = of t2] {$t_4$};
\node (t5) [right = of t2] {$t_5$};
\node (t6) [right = of t4] {$t_6$};
\tikzset{every node/.style={circle, draw=white, fill=white}}
\draw[special]  (t1) -- node {$1$} (t2);
\draw (t1) -- node {$0.8$} (t3);
\draw[special] (t1) -- node [near start] {$0.1$} (t4);
\draw[special] (t2) -- node [near start] {$1$} (t3);
\draw (t2) -- node {$0.3$} (t4);
\draw (t3) -- node {$0$} (t4);
\draw (t2) -- node {$0$} (t5);
\draw  (t4) -- node {$0.9$} (t6);
\draw[special]  (t5) -- node {$1$} (t6);
\draw (t2) -- node [near start] {$0$} (t6);
\draw [special] (t5) -- node [near start] {$1$} (t4);
\end{tikzpicture}
\vspace{-2mm}
\caption{\scriptsize Second iteration (after transforming cycle to path)}
\label{fig:ex3-3}
\end{subfigure}
\Description{Three graphs; the first has two disjoint cycles highlighted, the second has two disjoint paths highlighted, and the third has a single Hamiltonian path highlighted.}
\vspace{-3mm}
\caption{(Example~\ref{ex-3}) \Alg{CC} returns the ordering $\langle t_3, t_2, t_1, t_4, t_5, t_6 \rangle$}\label{fig:ex3}
\end{figure*}

\begin{example}\label{ex-3}
Consider a set of stories $\mathbf{t}=\set{t_1,\dots,t_6}$ with pairwise neutrality values as shown in the graph of Figure~\ref{fig:ex3-1} (for visual clarity, we omit four edges with weight zero).
\Alg{CC} starts by finding the max-weight cycle cover $\{(t_1, t_2, t_3), (t_4, t_5, t_6)\}$, as highlighted in the figure.
Then, it removes the min-weight edge from each cycle (Figure~\ref{fig:ex3-2}).
Next, the algorithm replaces these paths with super nodes: $\langle t_1, t_2, t_3\rangle$ and $\langle t_4, t_5, t_6\rangle$.
In the second iteration, the algorithm joins the two super nodes to form the cycle $(t_3, t_2, t_1, t_4, t_5, t_6)$ and removes the edge $(t_6, t_3)$ to create the final super node (Figure~\ref{fig:ex3-3}).
The neutrality of the resulting ordering under $\AggAvg$ aggregation is $4.1 / 5 = 0.82$.
\end{example}

\begin{theorem}\label{th:mtsp2}
\Alg{CC} returns a $1/2$-approximation for $\MiscProb{Path-}$-$\MiscProb{MaxTSP}$.
\end{theorem}

We can compute a max-weight bipartite matching in $O\bigl(n^3\bigr)$ time using the Hungarian method~\citep{10.1145/321694.321699, https://doi.org/10.1002/net.3230010206}. It can also be done in expected time $O\bigl(n^2 \log n\bigr)$ if the edge weights are i.i.d. random variables~\citep{https://doi.org/10.1002/net.3230100205}. Again, we can use the linear-time approximation algorithm for general max-weight matching instead.
Properly implemented, the runtime of each iteration of the loop is dominated by the cost of the bipartite matching. If we use the Hungarian method, then each iteration takes $O\bigl(|V_k|^3\bigr)$ time. By the master theorem, the overall runtime is then $O\bigl(n^3\bigr)$.

\tightsubsection{Approximation via 3-Cycle Cover}
So far, both algorithms proposed are $1/2$-approximation algorithms. Our third algorithm, \Alg{3CC} (pseudocode in the appendix) improves the approximation factor to $2/3$, but at a high computation cost.
It finds a max-weight 3-cycle cover, defined as a set of cycles of maximum total weight such that every vertex is included in exactly one cycle and every cycle has length at least 3.
It then removes the min-weight edge from each cycle and arbitrarily joins the resulting paths to form a Hamiltonian path.


We reduce the problem of computing a max-weight 3-cycle cover to that of computing a max-weight matching on a more complex graph. A thorough presentation of the reduction is given by \citet{stack8570}. The original reduction was a generalization of this argument given by \citet{tutte_1954}.

We have not yet explicitly specified a subroutine to compute a max-weight matching. We can use Edmonds' blossom algorithm or the linear time approximation algorithm for max-weight matching. The time to compute the matching dominates the rest of the computation, so the overall time complexity of \Alg{3CC} is the time complexity of running the preferred algorithm on a graph with
$|V| = 2n^2 - 4n$ and
$|E| = n^3 - 3.5n^2 + 2.5n$.
With the blossom algorithm, this leads to a overall runtime of $O\bigl(n^7\bigr)$. As such, this algorithm is not practical (we do not use it in our experiments), but it is of theoretical interest, as evidenced by the following theorem.

\begin{theorem}\label{th:mtsp3}
\Alg{3CC} returns a $2/3$-approximation for $\MiscProb{Path}$-$\MiscProb{MaxTSP}$.
\end{theorem}
\tightsection{Maximizing Neutrality under Min Aggregation}\label{sec:aggmin}
Now, we consider the scenario where our aggregation function $f$ is the $\AggMin$ function: it returns the smallest element in a totally ordered set. This is another well motivated aggregation function; if we are okay with many imperfect pairs but just want to make sure that no pair is \emph{too} bad, then it makes sense to maximize the neutrality of the most biased pair.
In the graph representation, the problem is now equivalent to finding a Hamiltonian path with maximum min-weight edge. This problem was first studied by \citet{arkin1999maximum}. We will refer to it as the ``path maximum scatter traveling salesman problem'', or $\MiscProb{PathMaxScatterTSP}$.
\begin{definition}[$\MiscProb{PathMaxScatterTSP}$]
Given a graph $G$, \emph{$\MiscProb{PathMax}$-$\MiscProb{ScatterTSP}$} is the problem of finding a Hamiltonian path with minimum edge weight maximized.
\end{definition}
The cycle variant has been successfully addressed with heuristic methods~\citep{8790018}, but there are no published results on algorithms for $\MiscProb{PathMaxScatterTSP}$.
Unfortunately, we will have to rely on heuristic methods, as we have the following inapproximability result by \citet{arkin1999maximum}.
\begin{theorem}\label{thm:minagg}
There is no polynomial-time constant-factor approximation algorithm for $\MiscProb{PathMaxScatterTSP}$ unless $\mathrm{P} = \mathrm{NP}$.
\end{theorem}
\begin{corollary}
There is no polynomial-time constant-factor approximation algorithm for $\NeutProb{\AggMin}$ unless $\mathrm{P} = \mathrm{NP}$.
\end{corollary}

We now adapt the $\MiscProb{BottleneckATSP}$ heuristic algorithm of \citet{LARUSIC201420} and design the first heuristic algorithm for the $\MiscProb{PathMaxScatterTSP}$ problem.

Given a graph $G = (V,E)$ and parameter $\delta$, we define the graph $G' = (V,E')$ with edge set defined as follows. For each edge $(u,v) \in E$ with weight at least $\delta$, we have an edge $(u,v) \in E'$ with weight 0. For each edge $(u,v) \in E$ with weight $w$ less than $\delta$, we have an edge $(u,v) \in E'$ with weight $\delta - w$. Then, there is a Hamiltonian cycle in $G$ with minimum edge weight at least $\delta$ if and only if there is a cycle in $G'$ with total weight 0. We use any heuristic solver for $\MiscProb{TSP}$ to predict whether such a cycle exists in $G'$; we will use 2-opt~\citep{10.2307/167074}, a simple but effective local search algorithm. One useful property of 2-opt is that it is an \emph{anytime} algorithm. The user can choose to terminate it before convergence and obtain a slightly sub-optimal solution to save time if computational resources are scarce.

If we perform a binary search over all possible edge weights $\delta$, we can estimate the maximum value of $\delta$ such that there is a Hamiltonian cycle in $G$ with minimum edge weight at least $\delta$. If we then remove the min-weight edge from that cycle, the resulting path is an approximate solution to the instance of $\MiscProb{PathMaxScatterTSP}$.
The full pseudocode is shown in \autoref{alg:PathMaxScatterTSP}.

\begin{algorithm}[bhpt]
	\caption{Heuristic for $\MiscProb{PathMaxScatterTSP}$}
    \label{alg:PathMaxScatterTSP}
	\begin{algorithmic}[1]
	\Procedure{IsFeasible}{$G = (V,E), \delta$}
	    \LongState{1}{Define the graph $G' = (V,E')$ such that for each edge $(u,v) \in E$ with weight $w$, we have an edge $(u,v) \in E'$ with weight $\max(\delta - w, 0)$.}
		\State $C' \gets \textsc{2-opt}(G')$
		\LongState{1}{Let $C$ be the corresponding cycle in $G$ and $W$ the total weight of $C'$.}
		\State \textbf{if} $W = 0$ \textbf{then} \textbf{return} $C, \texttt{True}$
		\State \textbf{else}  \textbf{return} $C, \texttt{False}$
	\EndProcedure
	\Procedure{BinarySearch}{$G,l,r$}
	    \State $m \gets \lceil(l+r)/2\rceil$
        \State $C, X \gets \textsc{IsFeasible}(G,m)$
        \State \textbf{if} $l = r$ \textbf{then} \textbf{return} $C$
        \State \textbf{if} $X$ \textbf{then} \textbf{return} $\textsc{BinarySearch}(G,m,r)$
        \State \textbf{return} $\textsc{BinarySearch}(G,l,m-1)$
	\EndProcedure
	\Procedure{PathMaxScatterTSP}{$G$}
        \State Consider the sorted list of all edge weights.
        \State Let $l$ and $r$ be the minimum and maximum, respectively.
        \State $C \gets \textsc{BinarySearch}(G,l,r)$
        \LongState{1}{Remove the min-weight edge from $C$ to construct a path $P$.}
        \State \textbf{return} $P$
	\EndProcedure
	\end{algorithmic}
\end{algorithm}
\tightsection{Experiments}\label{sec:experiments}
Now that we have finished introducing all the algorithms, we present our experiments.
First, we describe the data collection and generation process, then we discuss hardware and implementation details, and finally, we report the results of the experiments.
Our implementations and data are freely available online~\citep{rishi_advani_2023_7955163}.

\tightsubsection{Data}
To measure the empirical performance of our algorithms, we tested them on real, semi-synthetic, and synthetic data.

\stitle{Real Data}
%
We selected two news sources based in the United States: American Thinker and The Federalist.
On July 24, 2022, we collected the first 11 headlines from each homepage.
Within each source, every possible pair of headlines was labeled by 3 out of 6 total annotators according to whether or not they thought it may lead to significant opinion priming.
The following prompt was given to all annotators.

\vspace{1mm}
\noindent
\framebox[\linewidth]{\itshape
\parbox{0.95\linewidth}{
Suppose an average adult residing in the United States is viewing news headlines.\\
If the subject views headline A and headline B together, will their impression of either story likely be different from what it would have been if the subject had viewed them individually?\\
I.e., would viewing the headline of one story influence their opinion on the veracity of the content of the other story or the causes, effects, or benefits of the events discussed within?}
}
\vspace{1mm}

Every annotator was given 55 headline pairs; for each pair, 
the possible answers were ``yes'', ``no'', and ``maybe'', corresponding to pairwise neutrality values of $0$, $1$, and $0.5$, respectively.
For 35.5\% of the pairs, all annotators agreed. For 85.5\%, the majority agreed. The 15.5\% of cases where there was no consensus confirm our expectation that opinion priming can be specific to each individual.

After collecting the labels, we took the average value for each pair of headlines and used these values to define the POP function. These values serve as loose estimates of average audience perceptions.

\stitle{Semi-Synthetic Data}
To show that our algorithms scale to larger data, we created a dataset larger than we were able to collect labels for.
We considered several common probability distributions, and for each one, we computed both the parameters that best fit our data and the likelihood of the data given that distribution with those parameters. We then chose the distribution with greatest likelihood. In short, we found the distribution that best fit the labeled data. For our data, the best fit was a beta distribution.


To generate the semi-synthetic data, we create a complete graph, and for each edge, we sample a value from the chosen distribution for the weight. In this way, we are able to generate a graph corresponding to a dataset of any size that has edge weights matching the distribution of weights in the real data.

\stitle{Synthetic Data}
In addition to the semi-synthetic data, we also generated graphs with edge weights not drawn independently.
Intuitively, if $C(u,w)$ and $C(v,w)$ are high, then $C(u,v)$ is likely to be high as well, so in this setting, we draw edge weights from the original beta distribution, but then enforce that there are no triangles in the graph where only two edges have ``high'' POP values.
Adding a second distribution also served as a way to test the robustness of our algorithms to changes in the data distribution.

\tightsubsection{Implementation}
\stitle{Hardware}
All experiments were run on a machine with an Intel(R) Core(TM) i5-8265U CPU @ 1.80 GHz and 16.0 GB of memory running Windows 11 Pro.

\stitle{Software}
All methods were implemented using Python 3.10.2. The NetworkX package~\citep{SciPyProceedings_11} (version 2.6.3) was used for the graph representations and several graph algorithms. The SciPy library~\citep{2020SciPy-NMeth} (version 1.8.0) was used in the implementation of the statistical test and for the fitting of and sampling from probability distributions. 

\stitle{Implementation Details}
The NetworkX method used for max-weight matching relies on a blossom-type algorithm~\citep{edmonds_1965}. For max-weight bipartite matching, we used the algorithm of \citet{https://doi.org/10.1002/net.3230100205}.

When evaluating the algorithms for maximizing neutrality, the methods were run on 4 random graphs with edge weights sampled from the same distribution, and the neutrality values obtained were averaged; this helps account for the randomness in the data. In addition, for each graph, the methods were run 3 times and the minimum execution time was recorded to account for any unrelated changes in processor utilization affecting the speed of computation.

When evaluating the algorithm for detecting cherry-picked orderings, the methods were run 5 times and the minimum execution time was recorded; the relative efficiency of the detection methods allows us to run them a higher number of times.

\tightsubsection{Results}
\begin{figure}[pt]
\centering
\includegraphics[width=\linewidth]{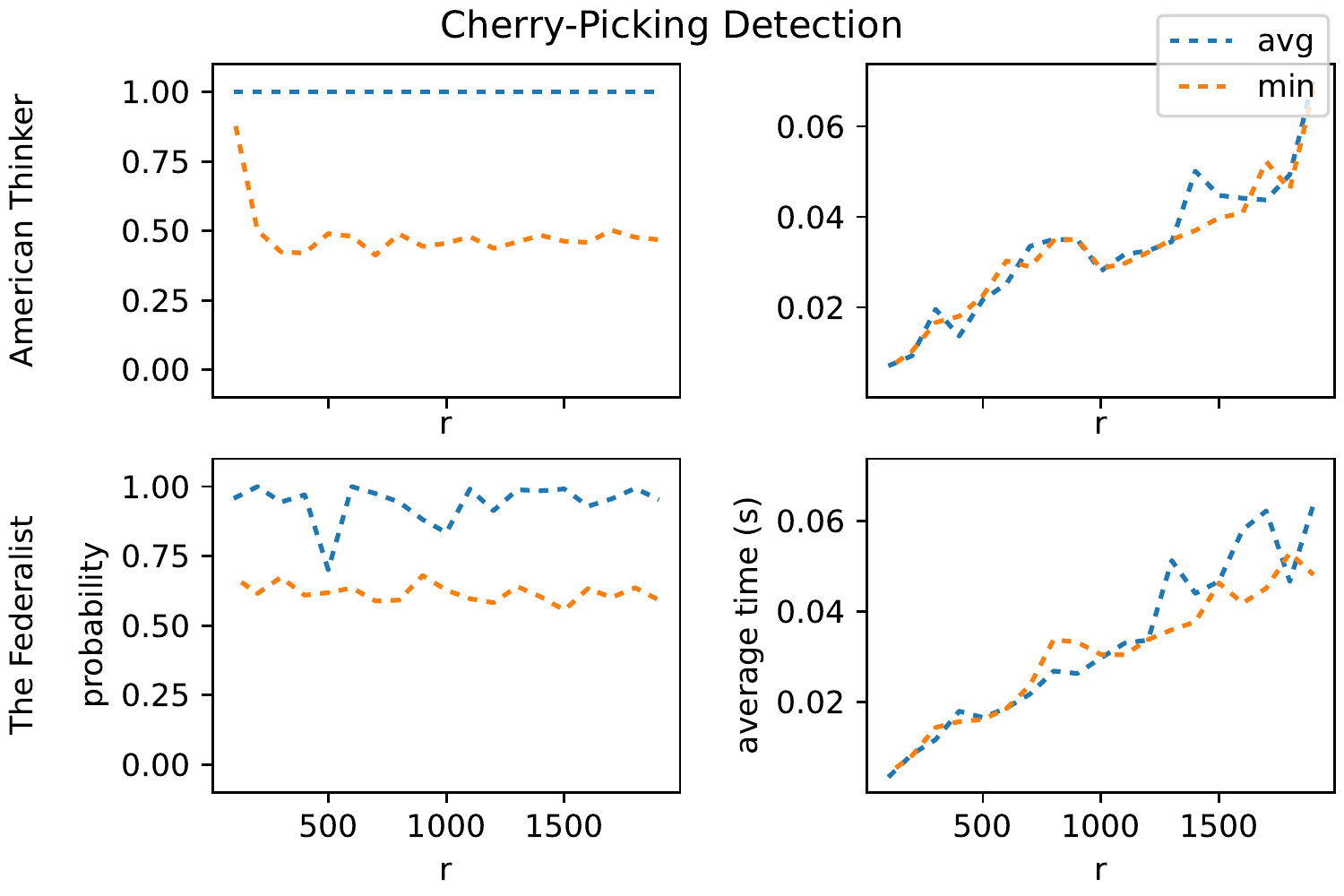}
\Description{A plot of probabilities and execution times varying with sample size for American Thinker and The Federalist; the plots of probabilities are roughly horizontal lines, and those of execution times are linear.}
\vspace{-8mm}
\caption{\small Left: upper bound on the probability of a random ordering having the neutrality of the given ordering; Right: average wall-clock execution time of detection method}
\label{fig:detection}
\end{figure}
\begin{figure*}[bhpt]
\begin{minipage}{0.22\linewidth}
\includegraphics[width=\linewidth]{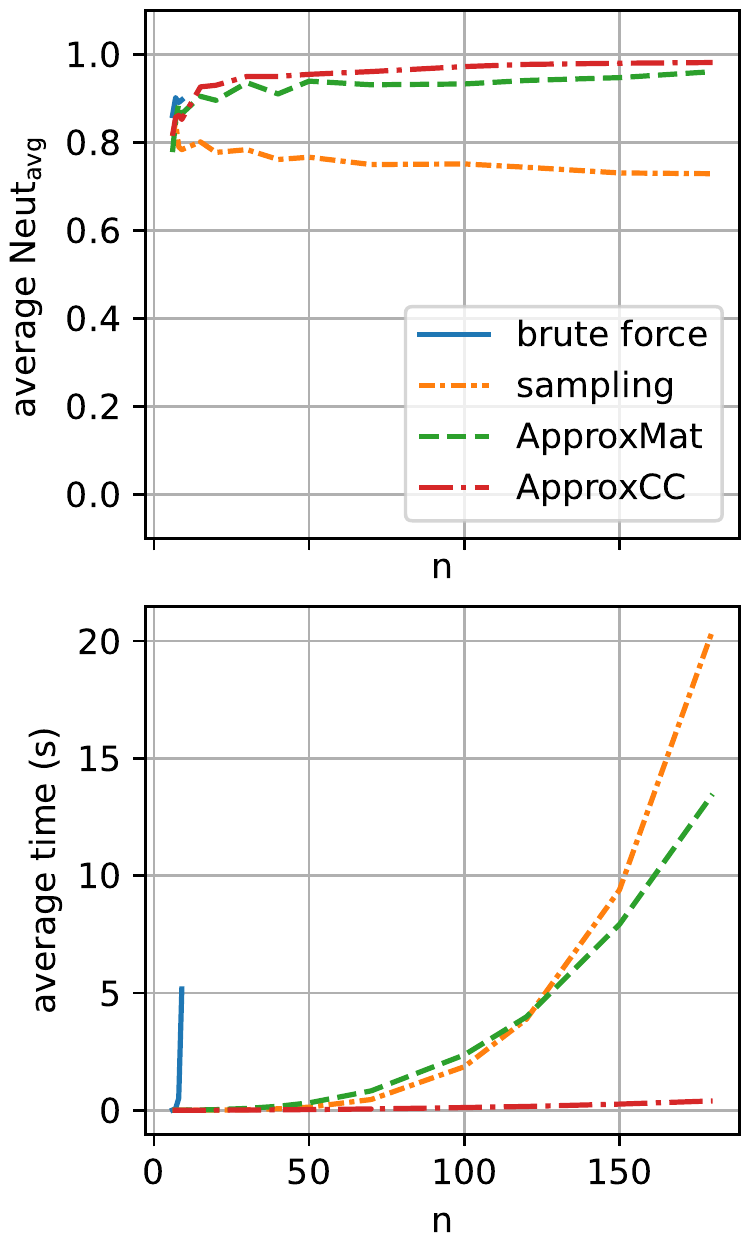}
\Description{A plot of neutrality and execution time. The neutrality of ApproxMat and ApproxCC converges to near 1 and that of the sampling baseline converges to near 0.75. The execution time is fastest for ApproxCC and slowest for the brute force method.}
\vspace{-8mm}
\caption{\small Neutrality and exec. time on semi-synthetic data ($\AggAvg$ aggregation)}
\label{fig:avgagg}
\end{minipage}
\hfill
\begin{minipage}{0.22\linewidth}
\includegraphics[width=\linewidth]{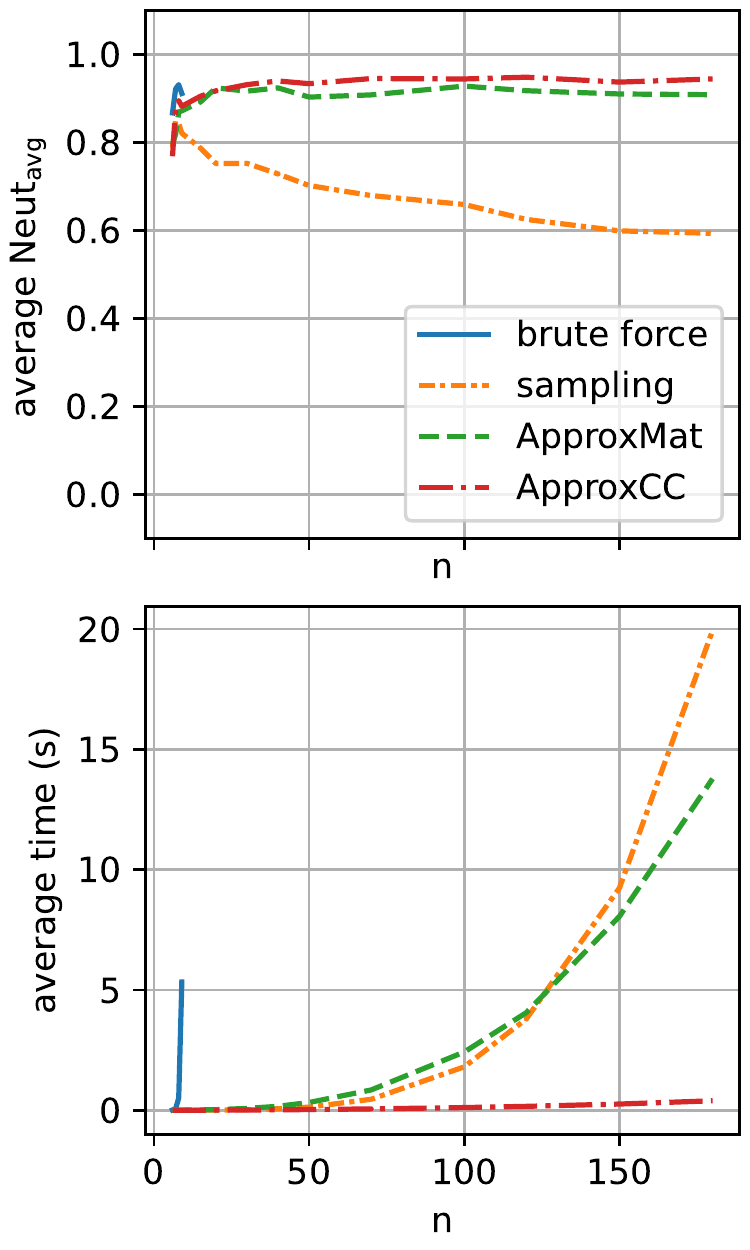}
\Description{A plot of neutrality and execution time. The neutrality of ApproxMat and ApproxCC converges to near 0.95 and that of the sampling baseline converges to near 0.6. The execution time is fastest for ApproxCC and slowest for the brute force method.}
\vspace{-8mm}
\caption{\small Neutrality and exec. time on synthetic data ($\AggAvg$ aggregation)}
\label{fig:avgagg2}
\end{minipage}
\hfill
\begin{minipage}{0.22\linewidth}
\includegraphics[width=\linewidth]{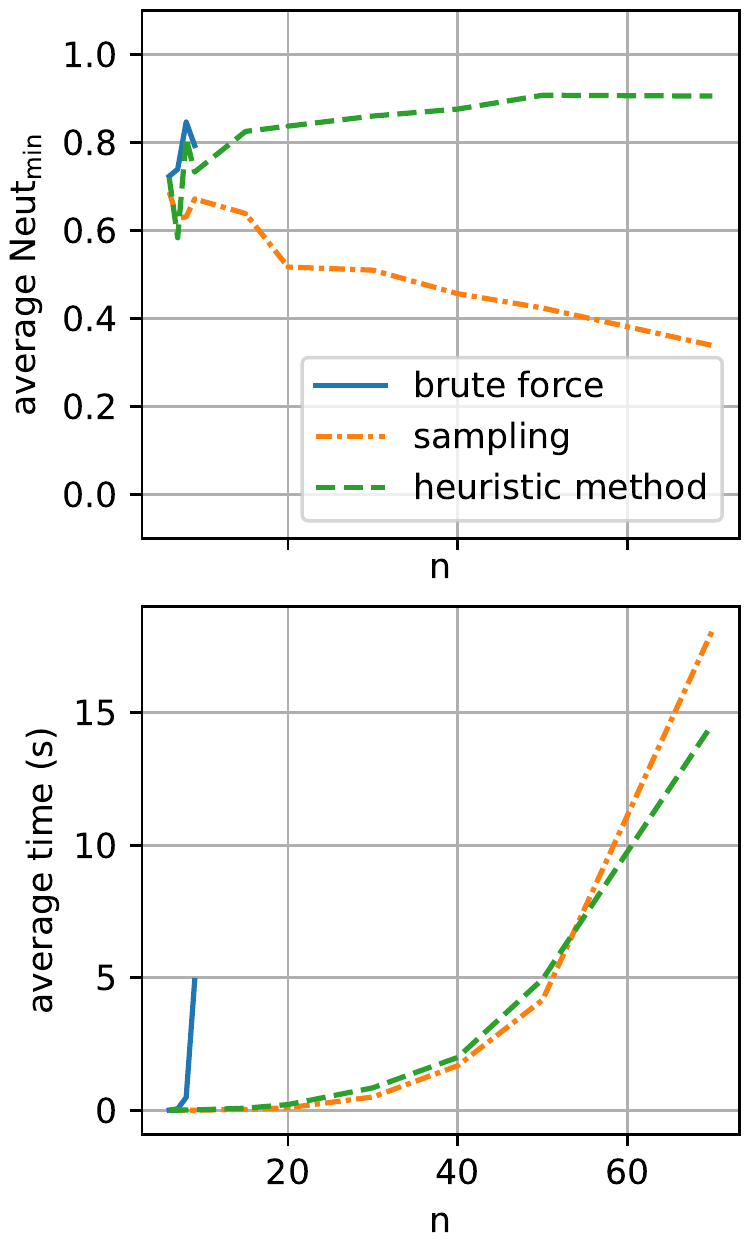}
\Description{A plot of neutrality and execution time. The neutrality of the heuristic method converges to near 0.9 and that of the sampling baseline steadily falls past 0.4. The execution time is fastest for the heuristic method and the sampling baseline and slowest for the brute force method.}
\vspace{-8mm}
\caption{\small Neutrality and exec. time on semi-synthetic data ($\AggMin$ aggregation)}
\label{fig:minagg}
\end{minipage}
\hfill
\begin{minipage}{0.22\linewidth}
\includegraphics[width=\linewidth]{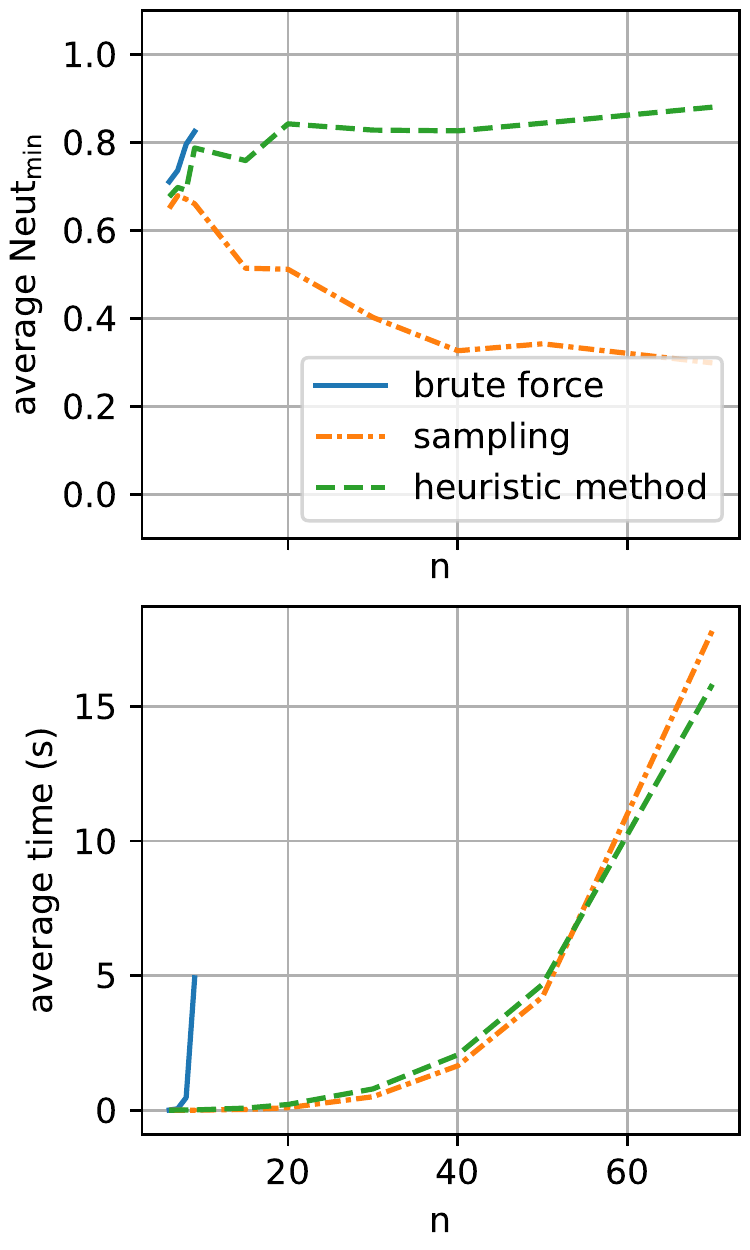}
\Description{A plot of neutrality and execution time. The neutrality of the heuristic method rises past 0.85 and that of the sampling baseline falls past 0.35. The execution time is fastest for the heuristic method and the sampling baseline and slowest for the brute force method.}
\vspace{-8mm}
\caption{\small Neutrality and exec. time on synthetic data ($\AggMin$ aggregation)}
\label{fig:minagg2}
\end{minipage}
\end{figure*}
\stitle{Detection}
To illustrate the process of detecting cherry-picked orderings, we run our algorithm with varying values of $r$. Recall that $r$ represents the number of random permutations sampled for the statistical test. Thus, higher values of $r$ lead to more accurate testing but slower computation. The values of $r$ used in the experiments are $\set{100, 200, \dots, 2000}$. The data in \autoref{fig:detection} confirms that the probabilities converge as $r$ grows.

For a conclusive test, in order to detect evidence of cherry-picking in the real dataset, we ran the detection algorithm with a much larger value of $50000$ for $r$ to get the tightest bounds.

Under $\AggAvg$ aggregation, we did not discover any significant evidence of cherry-picking by either source.
Under $\AggMin$ aggregation, on the other hand, \emph{we found evidence of potential cherry-picking for both sources}.
For American Thinker, we found that the probability that a random ordering would have neutrality as far from the mean as that of the true ordering is bounded above by $46.49\%$.
Likewise, for The Federalist, we obtained an upper bound of $59.91\%$.
This does not prove that the orderings were cherry-picked, or if they were, give proof of malicious intent, but it does give evidence that cherry-picking may have occurred, especially in the case of American Thinker.

In both cases, the computed neutrality of the ordering was \emph{\underline{zero}} under $\AggMin$ aggregation. We computed the maximum possible neutrality via the brute force method for comparison: $0.67$ for American Thinker\footnotemark{} and $0.83$ for The Federalist.

\footnotetext{Our heuristic method also found an optimal solution --- in $0.02$ seconds, instead of 8 minutes.}

The following is an example of a headline pair with neutrality zero from American Thinker:
\begin{itemize}[leftmargin=*]
\item ``The Merriam-Webster's online dictionary redefines `female'{}''
\item ``Crayola has joined the woke brigade with a vengeance''
\end{itemize}
In this case, viewing the second headline may prime the viewer to consider Merriam-Webster's actions to be ``woke'', a term that is often used as a pejorative. If they had viewed the first headline individually, they would be more likely to form their own (potentially less biased) opinion on the story.

\stitle{Maximizing Neutrality}
We used the algorithms from \S\ref{sec:aggsum} and \S\ref{sec:aggmin} to find orderings with high neutrality for the semi-synthetic and synthetic data. For the semi-synthetic data, the results for $\AggAvg$ aggregation are shown in \autoref{fig:avgagg} and the results for $\AggMin$ aggregation in \autoref{fig:minagg}. For the synthetic data, the results for $\AggAvg$ aggregation are shown in \autoref{fig:avgagg2} and the results for $\AggMin$ aggregation in \autoref{fig:minagg2}.

The values of $n$ used for testing are based on the sizes expected of real-life datasets. It would be unlikely for a reader to view a contiguous list of over 200 news headlines. Furthermore, given the computational complexity of the problem, few algorithms would be able to perform well far beyond that point. The values of $n$ used in our experiments are $\set{6, 7, 8, 9, 15, 20, 30, 40, 50, 70, 100, 120, 150, 180}$.
For $\AggMin$ aggregation, we stop at $n=70$.

We also tested a ``sampling'' baseline that computes the neutrality of many random orderings and selects the best one. To enable a fair comparison, under $\AggAvg$ aggregation, we set it to run for approximately the amount of time that \Alg{Mat} takes, and under $\AggMin$ aggregation, the amount of time that the heuristic method takes.
We can see that our algorithms perform significantly better than the sampling baseline. Next, the execution times confirm that \Alg{CC} is significantly faster than \Alg{Mat} ($O\bigl(n^3\bigr)$ vs.\ $O\bigl(n^4\bigr)$). Finally, we can see just how slow the brute-force method is --- it is infeasible to run it for $n > 9$ in the experiments.

In addition to confirming the theoretical time complexities, we learn from the results that \Alg{CC} performs slightly better than \Alg{Mat}. It is unclear why this is the case, but it seems to consistently compute slightly better orderings. We can also see that our algorithms are at or near optimal for small $n$. We cannot make any conclusions about their optimality for larger $n$ since we are unable to compute the solution by brute force for larger $n$, but we expect that they are near optimal.

The experimental results also show that our algorithms are robust to changes in the data distribution. Neither the neutrality or execution time changes significantly as a result of the change in distribution (whereas the sampling baseline performs much worse on the synthetic data).

\stitle{Early Stopping}
\begin{figure}[bhpt]
\begin{minipage}{0.47\linewidth}
\includegraphics[width=\linewidth]{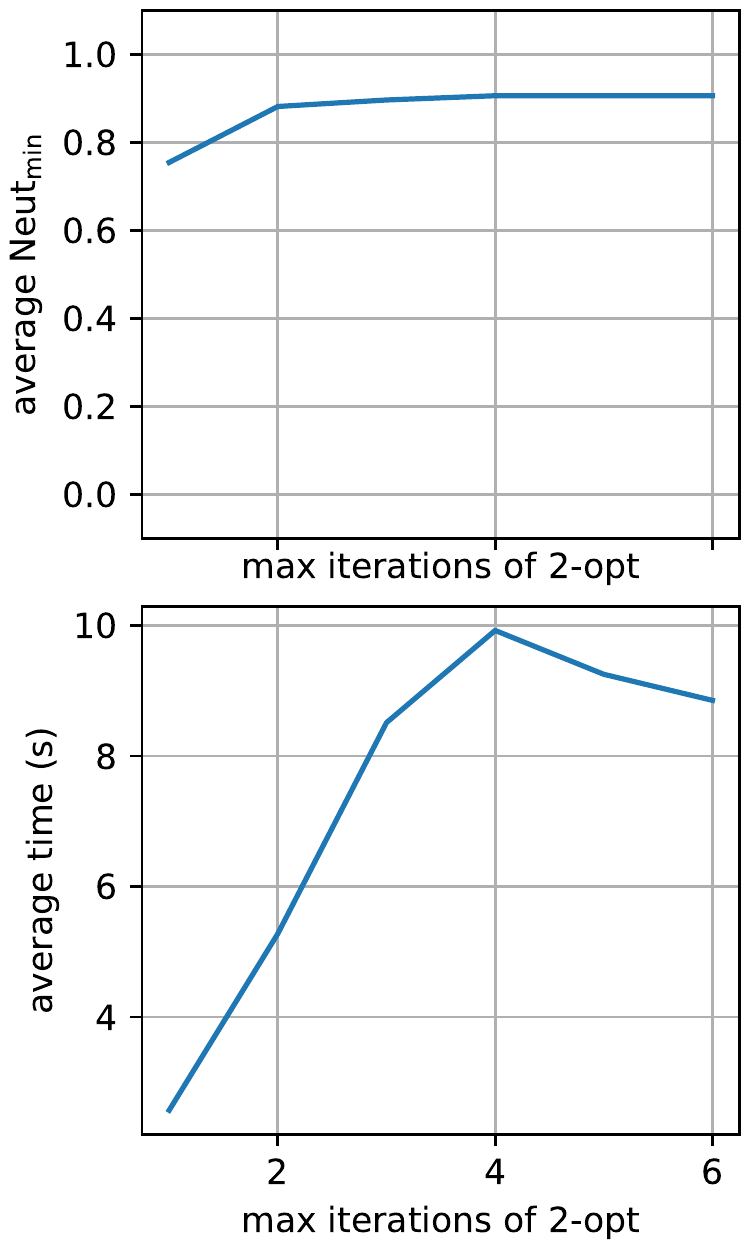}
\Description{A plot of neutrality and execution time. The neutrality converges near 0.9. The execution time is roughly linear until t=3.}
\vspace{-8mm}
\caption{\small Early stopping in 2-opt on semi-synthetic data}
\label{fig:early_stopping}
\end{minipage}
\hfill
\begin{minipage}{0.47\linewidth}
\includegraphics[width=\linewidth]{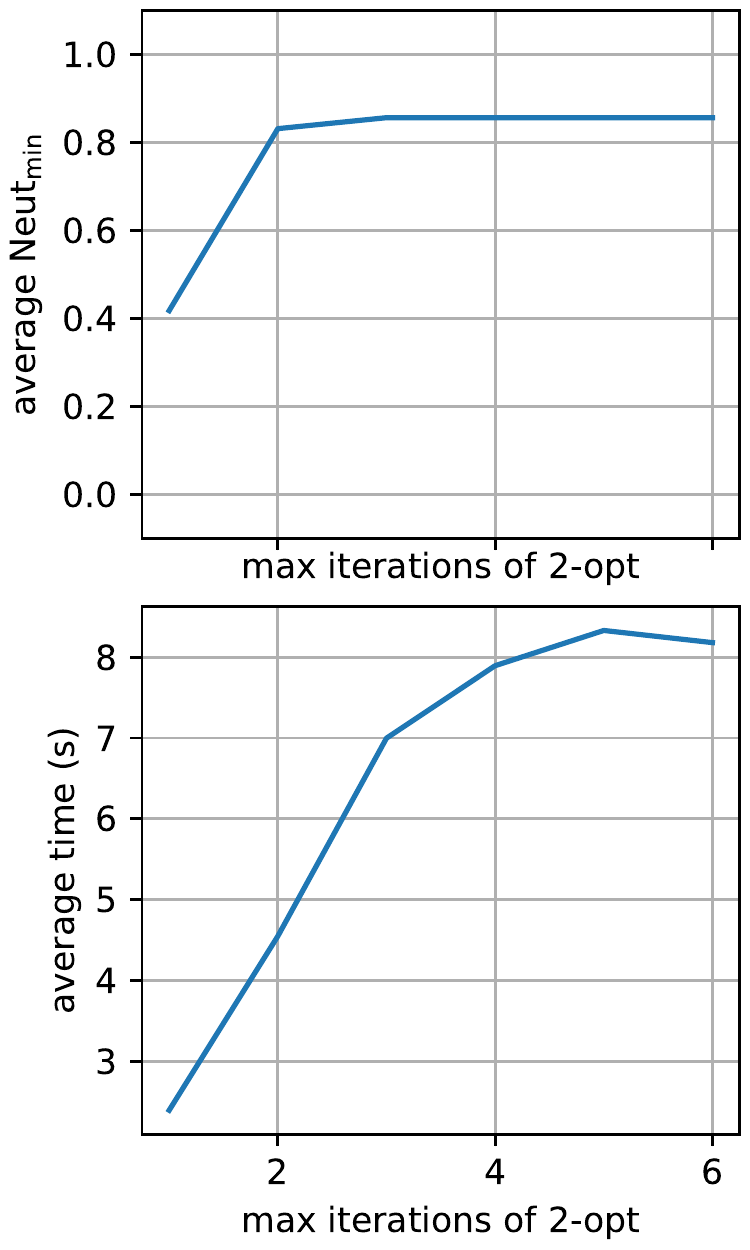}
\Description{A plot of neutrality and execution time. The neutrality converges near 0.85. The execution time is roughly linear until t=3.}
\vspace{-8mm}
\caption{\small Early stopping in 2-opt on synthetic data}
\label{fig:early_stopping2}
\end{minipage}
\end{figure}
Finally, we measured the effect of early stopping on the heuristic method. As mentioned earlier, the 2-opt subroutine has the ``anytime property'': it can be terminated at any time and still return a solution. For the other experiments, we simply ran 2-opt until it converged, but for this experiment, we ran it for a specified number of iterations (or terminated it if it converged early).

Using a fixed value of $n = 60$, for $1 \leq t \leq 6$, we ran the heuristic method but always terminated the 2-opt subroutine after at most $t$ iterations. The results are shown in Figures~\ref{fig:early_stopping}~and~\ref{fig:early_stopping2}.
The neutrality increases as the number of allowed iterations increases, but plateaus after about 3 or 4 iterations. However, the neutrality is still relatively high after just 2 iterations, with the benefit of a significant reduction in computation time.
\tightsubsection{Existence of Priming (User Study)}\label{sec:exp:userstudy}
The occurrence of priming has been well studied and has been analyzed in many related settings, including news consumption~\citep{Baumgartner12}, social networks~\citep{Agadjanian22}, job interviews~\citep{Simonsohn13}, crowdsourcing~\citep{Draws21}, and annotation~\citep{Damessie18}.
To verify that priming does indeed occur in our setting of viewing news headlines, we ran a user study (n=59).

We presented a test group and a control group (formed via convenience sampling) with a set of 9 fictional news headlines, and afterwards, we asked them their opinions of several people involved in the stories (``very negative'', ``negative'', ``positive'', or ``very positive'').
The full set of headlines and images of the survey interface can be found in the appendix.
The test group had the following pair of headlines placed next to each other in the ordering ($|s_j - s_i| = 1$); the control group had them separated ($|s_j - s_i| = 5$).
\begin{itemize}[leftmargin=*]
    \item ``City's high school graduation rates at lowest in decades''
    \item ``High school principal celebrates 10 years''
\end{itemize}
When surveyed at the end, after removing 6 responses that failed attention checks, 39\% of the participants in the test group had formed a negative impression of the principal, compared to 16\% in the control group.
This difference is statistically significant (Boschloo's exact test, $\mathbf{p=0.0337}$).

\balance

\tightsection{Related Work}\label{sec:related}
To the best of our knowledge, this paper is the first to study the effect of bias in news ordering. However, there are several areas of related work that we discuss next.

\stitle{Media Bias}
Neutrality in the \emph{ordering} of news headlines is the focus of this paper.
This is only one aspect of media bias, a large socio-technical problem with many dimensions~\citep{hamborg2019automated}.
Among many other facets, a well-recognized component of media bias is selection bias~\citep{hocke1996determining, bourgeois2018selection, kozman2022selection, lazaridou2016identifying}.
In general, selection bias happens when the data selection process does not involve proper randomization~\citep{winship1992models, heckman1990varieties, shahbazi2022survey}.
A study by \citet{bourgeois2018selection} uses the predictability of the news coverage to measure selection bias.

Diversifying search results~\citep{agrawal2009diversifying, drosou2017diversity} has been considered in efforts to reduce media bias~\citep{tintarev2018same, lunardi2019representing, giannopoulos2015algorithms, kitchens2020understanding}.
In particular, content spread in online social networks is affected by social bubbles and echo chambers~\citep{vaccari2021outside, nikolov2015measuring, dubois2018echo, cinelli2021echo, bruns2019filter}, which significantly bias the spread of information.
Diversifying news recommendations~\citep{kitchens2020understanding, tintarev2018same, gharahighehi2021diversification, lunardi2019representing, desarkar2014diversification, bauer2021diversification} has been an effective technique for breaking echo chambers.

\stitle{Computational Fact Checking}
Whereas our work emphasizes the importance of the ordering of news stories, prior work in the literature focuses on the veracity of the content of said stories~\citep{guo2022survey, zhou2020survey, NakovCHAEBPSM21}. There have been remarkable advancements in fake news detection in the past decade~\citep{cohen2011computational, wu2014toward, hassan2017toward, hassan2017claimbuster, asudeh2020detecting, hassan2015detecting, wu2017computational, hassan2014data}.
Early fake news detection efforts include manual methods based on expert domain knowledge and crowdsourcing~\citep{hassan2017toward, hassan2015detecting}.
Computational fact checking has since emerged, enabling automatic evaluation of claims~\citep{NakovCHAEBPSM21, guo2022survey, zhou2020survey}. 
These techniques heavily rely on natural language processing~\citep{li2007danalix, li2006constructing}, information retrieval~\citep{doan2012principles}, and graph theory~\citep{cohen2011computational}.
Related work includes knowledge extraction from different data sources~\citep{pawar2017relation, dong2014knowledge, grishman2015information}, data integration~\citep{steorts2016bayesian, magdy2010web, altowim2014progressive}, and credibility evaluation~\citep{esteves2018belittling, dong2015knowledge}.
A bulk of the recent techniques used in fake news detection are based in supervised learning~\citep{reis2019supervised, perez2017automatic, kaliyar2020fndnet}.
An increasing number of approaches are putting emphasis on the role of structured data~\citep{sumpter2021preserving, anadiotis2021empowering, saeed2021factchecking, chai2021tfv, asudeh2021perturbation, trummer2021webchecker}, as reflected in a special issue of the Data Engineering Bulletin~\citep{dbe21}.

\stitle{Fair Ranking}
Ranking news stories has been studied in the literature~\citep{del2005ranking, mccreadie2010news, tatar2014popularity}, but to the best of our knowledge, none of the existing work considers bias and neutrality in news ordering.
Fair ranking is a recent line of work that studies ordering a set of items or individuals to satisfy some fairness constraints~\citep{asudeh2019designing, biega2018equity, singh2018fairness, pitoura2021fairness}.
At a high level, existing work is divided into score-based ranking~\citep{zehlike2022fairness1, asudeh2019designing, asudeh2020fairly, asudeh2018obtaining} and learning-to-rank and recommender systems~\citep{zehlike2022fairness, gao2020toward, geyik2019fairness, beutel2019fairness}.
Despite the similarity in name, none of the existing work in this area can map onto our formulation of news ordering neutrality and is thus not useful in solving the problem proposed in this paper.

\stitle{Traveling Salesman Problem}
The (cycle) maximum traveling salesman problem has been studied since at least 1979~\citep{fisher_analysis_1979}. A survey on the maximum traveling salesman problem is given by \citet{barvinok_maximum_2007}, and the current state-of-the-art solution for it has a constant factor of 4/5~\citep{10.1007/978-3-319-59250-3_15}. On the other hand, the path maximum traveling salesman problem, which we introduce in this paper to model $\AggAvg$ aggregation, is not present in the literature to the best of our knowledge.
\tightsection{Final Remarks}\label{sec:discussion}
As this paper is opening up a new line of research in the fight against misinformation, there are numerous directions to explore.

\stitle{Data Collection}
A major challenge when evaluating news ordering neutrality is the collection of labels for each pair of news stories for constructing the POP function. One straightforward way is to survey the perceptions of the audience itself. However, it is not always possible to gain access to the audience's beliefs. 
The two main alternatives are crowdsourcing the labeling or having a domain expert provide the labels. The former is easy to scale but may result in inaccurate labels, while the latter results in accurate labels but is difficult to scale to large datasets.
One promising potential approach to alleviate the difficulties of the data collection process is to train large language models to classify pairs of news stories.

\stitle{Introducing Utility}
While maximizing neutrality is important, a corporation's main goal is to maximize profits. It would be an interesting research direction to introduce a notion of utility and attempt simultaneous maximization of neutrality and utility.

\stitle{Maximizing Neutrality}
In this work, we restricted ourselves to a simple decay function to represent the problem in the form of a graph. One natural extension is to study the problem with more complex decay functions.
For example, in layouts with multiple pages, it would be plausible to have a decay function with value 1 if two headlines are on the same page and 0 otherwise.

Another challenge is to try to find adversarial examples for the proposed algorithms. While we have proved approximation guarantees, we have not shown that they are tight. If adversarial examples are found, these bounds will be shown to be tight.

Finally, most of the algorithms for maximizing neutrality do not make any assumptions on the distribution of the data.
It would be interesting to see if better guarantees or algorithms can be discovered for specific data distributions.


\begin{acks}
This work has been supported in parts by the UIC University Fellowship, the \grantsponsor{nsf}{National Science Foundation}{https://www.nsf.gov/} (Grant No.~\grantnum{nsf}{2107290}), the \grantsponsor{anr}{ANR}{https://anr.fr/en/} project ATTENTION (\grantnum{anr}{ANR-21-CE23-0037}), and the Google Research Scholar Award.
We would like to thank the reviewers for their constructive feedback and \textbf{Dr.~H.\,V.~Jagadish} for his invaluable contribution in problem identification.
\end{acks}

\clearpage
\bibliographystyle{ACM-Reference-Format}
\bibliography{main}

\clearpage
\appendix
\tightsection{Directionality}
If we model the POP function as being sensitive to the order of the two stories taken as input, all but one of the results in this paper still hold.
\Alg{3CC} cannot be adapted for the directed setting because it relies on finding a max-weight 3-cycle cover and doing this is NP-hard in a directed graph~\citep[GT13]{gareyandjohnson}.
Fortunately, this algorithm is of theoretical interest only, and all other algorithms can be easily adapted, as-is, for the directed setting.
\tightsection{Pseudocode}
The pseudocode for approximately solving $\MiscProb{PathMaxTSP}$ via iterated matching is shown below.

\begin{algorithmic}[1]
\Procedure{ApproxMat}{$G=(V,E)$}
    \State $k \gets 0$;
     $G_k \gets G$;
     $size \gets n$
    \While{$size > 1$}
        \State Compute a max-weight matching $M$ in $G_k$.
        \State Construct a graph $G_{k+1}=(V_{k+1},E_{k+1})$ as follows.
        \ForAll{$((u_1,\dots,u_a), (v_1,\dots,v_b)) \in M$}
            \LongState{3}{Consider the edge $e$ in $E$ with highest weight from the set $\set{(u_1,v_1), (u_1,v_b), (u_a,v_1), (u_a,v_b)}$.}
            \State \textbf{if} $e = (u_1,v_1)$\textbf{:} Add $(u_a, \dots, u_1, v_1, \dots, v_b)$ to $V_{k+1}$.
            \State \textbf{if} $e = (u_1,v_b)$\textbf{:} Add $(u_a, \dots, u_1, v_b, \dots, v_1)$ to $V_{k+1}$.
            \State \textbf{if} $e = (u_a,v_1)$\textbf{:} Add $(u_1, \dots, u_a, v_1, \dots, v_b)$ to $V_{k+1}$.
            \State \textbf{if} $e = (u_a,v_b)$\textbf{:} Add $(u_1, \dots, u_a, v_b, \dots, v_1)$ to $V_{k+1}$.
        \EndFor
        \State If $|V_k|$ is odd, add the remaining vertex in $V_k$ to $V_{k+1}$.
        \ForAll{pairs $(u_1,\dots,u_a), (v_1,\dots,v_b)$ in $V_{k+1}$}
            \LongState{3}{Let $w$ be the weight of the edge in $E$ with highest weight from the set $\{\,(u_1,v_1), (u_1,v_b), (u_a,v_1), (u_a,v_b)\,\}$.}
            \LongState{3}{Add the edge $((u_1,\dots,u_a), (v_1,\dots,v_b))$ to $E_{k+1}$ with weight $w$.}
        \EndFor
        \State $size \gets |V_{k+1}|$;
        $k \gets k+1$
    \EndWhile
    \LongState{1}{\textbf{return} the Hamiltonian path $v_1, \dots, v_n$ in $G$ where $(v_1, \dots, v_n)$ is the sole vertex in $V_k$.}
\EndProcedure
\end{algorithmic}

The pseudocode for approximately solving $\MiscProb{PathMaxTSP}$ via 3-cycle cover is shown below.

\begin{algorithmic}[1]
\Procedure{Approx3CC}{$G = (V,E)$}
    \State Construct a graph $G' = (V',E')$ as follows.
    \LongState{1}{Let $w_{\max}$ be the weight of the edge with max weight in $E$.}
    \LongState{1}{For each vertex $v \in V$, add a complete bipartite graph $G_v = K_{n-1, n-3}$ to $G'$ with each edge having weight $w_{\max}$.}
    \LongState{1}{Let $G_{v,R}$ denote the side of the bipartition with $n-1$ vertices.}
    \LongState{1}{For each edge $(u,v) \in E$ with weight $w$, add an edge with weight $w$ between a pair of vertices from $G_{u,R}$ and $G_{v,R}$ such that for all $x \in V$, each vertex in $G_{x,R}$ has degree exactly $n-2$.}
    \LongState{1}{Compute a max-weight matching in $G'$. Let $C$ be the corresponding 3-cycle cover in $G$.}
    \State For each cycle $c \in C$, remove the min-weight edge.
    \State Arbitrarily join the paths together.
    \State \textbf{return} the resulting Hamiltonian path.
\EndProcedure
\end{algorithmic}
\tightsection{Binary Average Aggregation}\label{sec:bin_avg}
If we add the condition that the POP function, $C$, takes values in $\set{0,1}$, we can prove additional results. This can be interpreted as assuming that any pair of stories has either zero risk of giving rise to opinion priming or is absolutely certain to do so.

\begin{definition}[$\MiscProb{PathMaxTSP(0,1)}$]
Given a graph $G$ such that all edges have binary weight, \emph{$\MiscProb{PathMaxTSP(0,1)}$} is the problem of finding a Hamiltonian path with maximum total weight.
\end{definition}
We also introduce the cycle variant, a problem not yet addressed in the literature:
\begin{definition}[$\MiscProb{MaxTSP(0,1)}$]
Given a graph $G$ such that all edges have binary weight, \emph{$\MiscProb{MaxTSP(0,1)}$} is the problem of finding a Hamiltonian cycle with maximum total weight.
\end{definition}

\begin{lemma}\label{lem:1}
$\MiscProb{PathMaxTSP(0,1)}$ is NP-hard.
\end{lemma}
\begin{proof}
Consider a graph $G = (V,E)$ with the property that all edges have binary weight. Then, define $G' = (V,E')$ to be the unweighted graph with the same vertex set and an edge between any two vertices $u$ and $v$ if and only if $(u,v)$ has weight 1 in $E$.

If a solution to $\MiscProb{PathMaxTSP(0,1)}$ in $G$ has weight $n-1$, then there must exist a Hamiltonian path in $G'$. If a solution to $\MiscProb{PathMaxTSP(0,1)}$ in $G$ has weight less than $n-1$, then there cannot exist a Hamiltonian path in $G'$. Given a solution to $\MiscProb{PathMaxTSP(0,1)}$, we can decide $\MiscProb{HamPath}$ in constant time.

Suppose for contradiction that $\MiscProb{PathMaxTSP(0,1)}$ is not NP-hard. Then, by our previous claim, $\MiscProb{HamPath}$ is not NP-hard. This is a contradiction, since $\MiscProb{HamPath}$ is known to be NP-hard. Therefore, $\MiscProb{PathMaxTSP(0,1)}$ must be NP-hard.
\end{proof}

\begin{proof}[Proof of \autoref{thm:pathmaxtspnphard}]
$\MiscProb{Path}$ $\MiscProb{MaxTSP(0,1)}$ is a special case of $\MiscProb{PathMaxTSP}$. Thus, $\MiscProb{PathMaxTSP}$ is also NP-hard.
\end{proof}

Since $\MiscProb{PathMaxTSP(0,1)}$ is NP-hard, we seek an approximation algorithm for the problem. The simplest approach is to reduce it to another problem that already has a solution. The cycle variant of this problem, $\MiscProb{MaxTSP(0,1)}$, which we introduced in this paper, has not yet been studied, but its generalization, $\MiscProb{MaxTSP}$, has been addressed in the literature; the current state-of-the-art algorithm has an approximation factor of $4/5$~\citep{10.1007/978-3-319-59250-3_15}.

\begin{theorem}
\label{thm:PathMaxTSP(0,1)}
Given an $\alpha$-approximation for $\MiscProb{MaxTSP(0,1)}$, we can compute an $\alpha$-approximation for $\MiscProb{PathMaxTSP(0,1)}$ in $O(n)$ time.
\end{theorem}
\begin{proof}
Given a cycle that is an $\alpha$-approximation for $\MiscProb{MaxTSP(0,1)}$, remove the min-weight edge. If it has weight 1, the cycle has weight $n$ and the resulting path has weight $n-1$, which is the maximum weight possible for a path of length $n-1$, so it must be an optimal solution. Henceforth, assume that the min-weight edge has weight 0. Let $C$ be the weight of the cycle and $P=C$ the weight of the resulting path.

Let $C^*$ be the weight of an optimal solution to $\MiscProb{MaxTSP(0,1)}$. If we remove an edge, we have a path of weight at most $C^*$. If there was a path with greater weight, we could join the endpoints to form a cycle with weight greater than $C^*$, so the first path must be optimal. Let $P^*$ be its weight.
Then, we have
$\frac{P}{P^*} = \frac{C}{P^*} >= \frac{C}{C^*} >= \alpha$.

Thus, the path we have constructed gives an $\alpha$-approximation for $\MiscProb{PathMaxTSP(0,1)}$.

It takes $O(n)$ time to find and remove the min-weight edge, so we have constructed the approximation in $O(n)$ time.
\end{proof}
\begin{corollary}
There is a $4/5$-approximation algorithm for $\MiscProb{Path}$ $\MiscProb{MaxTSP(0,1)}$.
\end{corollary}
\begin{proof}
There is a $4/5$-approximation algorithm for $\MiscProb{MaxTSP(0,1)}$. Thus, by \autoref{thm:PathMaxTSP(0,1)}, we can construct a $4/5$-approximation for $\MiscProb{PathMaxTSP(0,1)}$.
\end{proof}

\balance

\tightsection{Average Aggregation Proofs}\label{sec:avg_proofs}
\begin{proof}[Proof of Theorem~\ref{th:mtsp1}]
We will show that the first iteration of \Alg{Mat} gives us a $1/2$-approximation. All future iterations cannot worsen the approximation, so we do not have to consider their effects.
Let $e_1, e_2, \dots, e_{n-1}$ be the sequence of edges in some solution to an instance of $\MiscProb{PathMaxTSP}$. Let $W$ be the total weight of the path.
The following sets of edges are matchings in the graph:
\begin{equation*}
    \set{e_1, e_3, \dots},
    \set{e_2, e_4, \dots}
\end{equation*}
Consider the matching of greater total weight (breaking a tie arbitrarily). It must have weight at least $W/2$. Thus, the max-weight matching in the graph must have weight at least $W/2$. Finally, if we arbitrarily patch the edges together to form a Hamiltonian path, the resulting path also has weight at least $W/2$. This gives us a $1/2$-approximation.
\end{proof}

\begin{proof}[Proof of Theorem~\ref{th:mtsp2}]
We show that the first iteration of \Alg{CC} gives us a $1/2$-approximation. All future iterations cannot worsen the approximation, so we do not have to consider their effects.
Let $W$ be the weight of a max-weight Hamiltonian path in $G$.
First, note that a max-weight Hamiltonian cycle must have weight at least $W$. Next, a Hamiltonian cycle is a cycle cover with 1 cycle, so the max-weight cycle cover constructed in the algorithm must have weight at least $W$.

Each cycle, by definition, has at least 2 edges. If we remove the min-weight edge from each cycle, we are removing at most half of the weight of each cycle. The resulting set of paths has total weight at least $W/2$. Accordingly, the constructed Hamiltonian path has weight at least $W/2$. This gives us a $1/2$-approximation.
\end{proof}

\begin{proof}[Proof of Theorem~\ref{th:mtsp3}]
Assume $n \geq 3$ (it is trivial to solve $\NeutProb{\AggAvg}$ if $n < 3$).
Let $W$ be the weight of a max-weight Hamiltonian path in $G$.
First, note that a max-weight Hamiltonian cycle must have weight at least $W$. Next, a Hamiltonian cycle is a 3-cycle cover with 1 cycle, so the max-weight cycle cover constructed in \Alg{3CC} must have weight at least $W$.

Each cycle, by construction, has at least 3 edges. If we remove the min-weight edge from each cycle, we are removing at most one third of the weight of each cycle. The resulting set of paths has total weight at least $2W/3$. Accordingly, the constructed Hamiltonian path has weight at least $2W/3$. This gives us a $2/3$-approximation.
\end{proof}
\tightsection{User Study Headlines}
\begin{itemize}[leftmargin=*]
    \item ``Mayor accused of mishandling city funds''
    \item ``Library temporarily closing doors for renovations''
    \item ``Navy veteran receives Silver Star Medal''
    \item ``City's high school graduation rates at lowest in decades''
    \item ``Residents urged to stay in their homes as temperatures reach -20s''
    \item ``Local artist's work featured in popular downtown bar''
    \item ``New bus routes added to accommodate increase in commuters''
    \item ``Man charged with DUI after crashing into restaurant''
    \item ``High school principal celebrates 10 years''
\end{itemize}
\tightsection{User Study Interface}

\begin{figure}[bhpt]
\includegraphics[width=\linewidth]{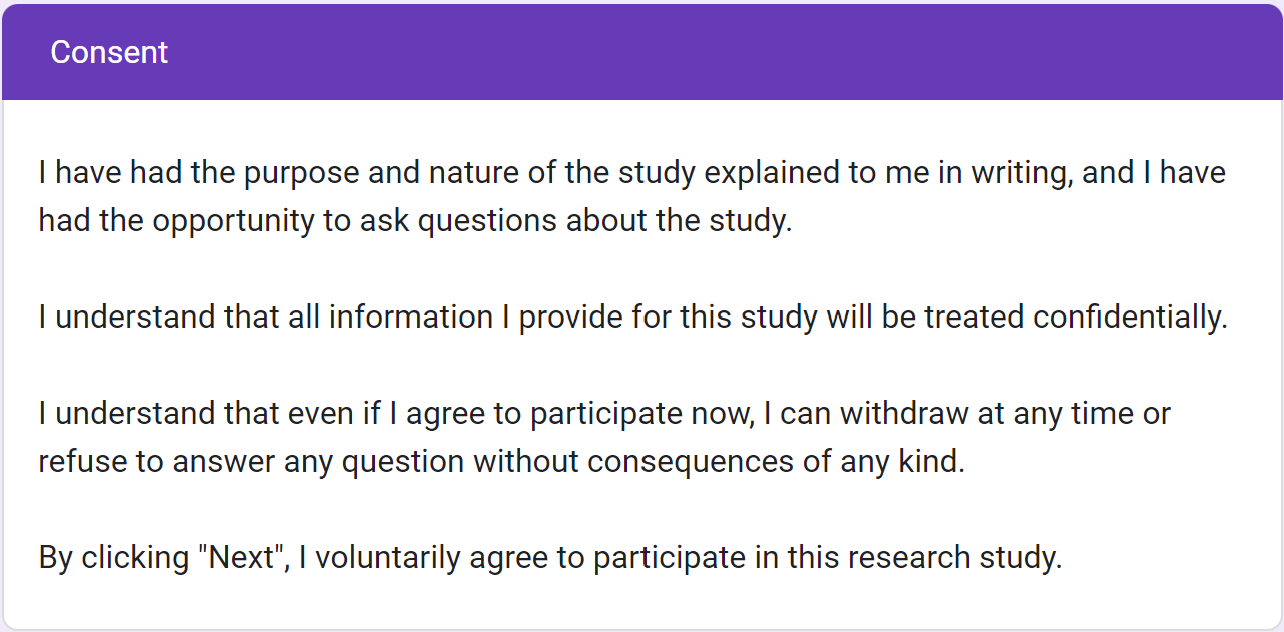}
\Description{A survey section asking the participant for consent.}
\vspace{-5mm}
\caption{\small Consent page of the user study}
\label{fig:user_study_1}
\end{figure}

\begin{figure}[bhpt]
\includegraphics[width=\linewidth]{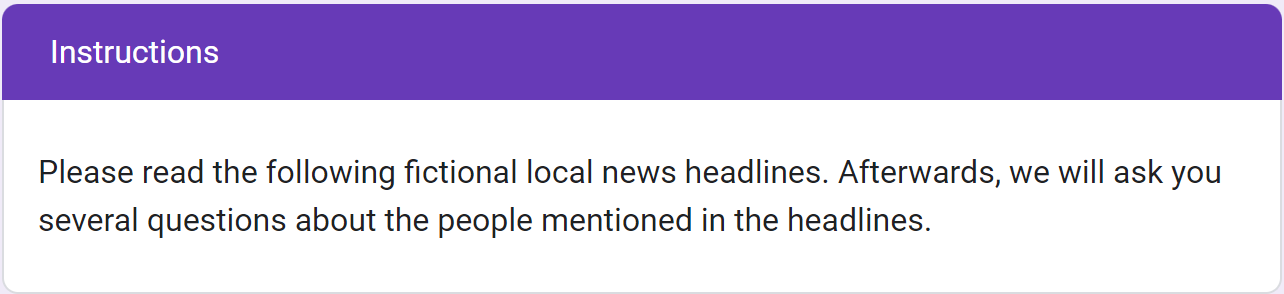}
\Description{A survey section providing instructions.}
\vspace{-5mm}
\caption{\small Instructions for the user study}
\label{fig:user_study_2}
\end{figure}

\begin{figure}[bhpt]
\includegraphics[width=\linewidth]{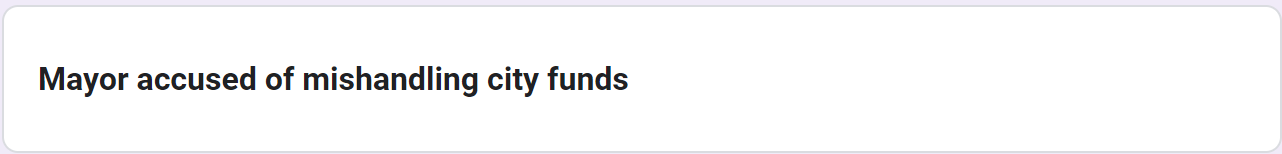}
\Description{A survey section containing a fictional news headline.}
\vspace{-5mm}
\caption{\small Example of a headline from the user study}
\label{fig:user_study_3}
\end{figure}

\begin{figure}[bhpt]
\includegraphics[width=\linewidth]{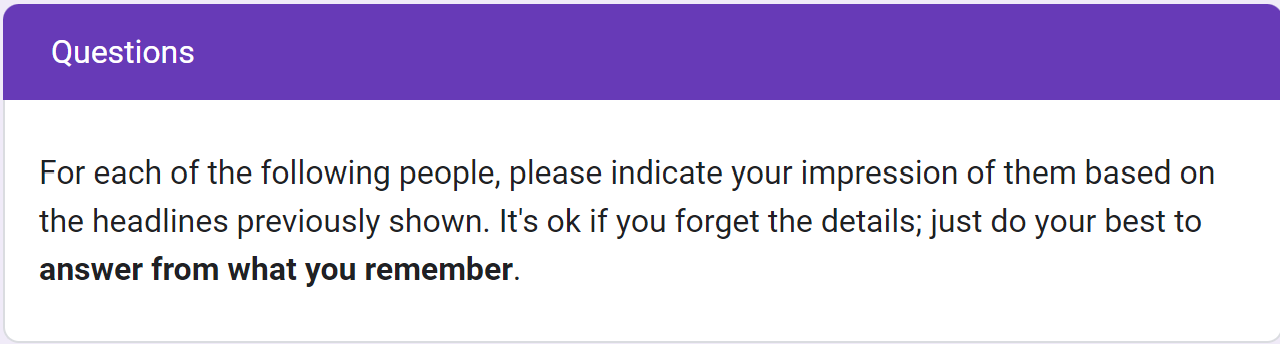}
\Description{A block of text in a survey giving instructions on how to answer questions.}
\vspace{-5mm}
\caption{\small Prompt from the user study}
\label{fig:user_study_4}
\end{figure}


\end{document}